\documentclass[hidelinks,onefignum,onetabnum]{siamart220329}
\usepackage{enumerate}
\usepackage{enumitem}
\usepackage{csquotes}
\usepackage{amssymb}

\definecolor{ForestGreen}{RGB}{34,139,34}

\newcommand{\seb}[1]{{\leavevmode\color{black}#1}}
\newcommand{\sebcancel}[1]{{\leavevmode\color{black}#1}}

\newcommand{\ben}[1]{{\leavevmode\color{black}#1}}
\newcommand{\keith}[1]{{\leavevmode\color{black}#1}}

\def\det{\mbox{\bf det }}
\usepackage{amsmath}

\newtheorem{assumption}{Assumption}

\def\interior{\operatorname{int}}
\def\bs{\boldsymbol}

\def\mcal{\mathcal}

\def\mcd{\mathcal{D}}
\def\mcn{\mathcal{N}}
\def\mcl{\mathcal{L}}

\def\I{\text{I}}
\def\S{\text{S}}
\def\DFE{\text{DFE}}

\DeclareMathOperator{\sgn}{sign}
\DeclareMathOperator{\spec}{spec}
\def\qed{\hfill $\Box$}

\usepackage{subfig}
\usepackage{lipsum}
\usepackage{amsfonts}
\usepackage{graphicx}
\usepackage{epstopdf}
\usepackage{algorithmic}
\ifpdf
  \DeclareGraphicsExtensions{.eps,.pdf,.png,.jpg}
\else
  \DeclareGraphicsExtensions{.eps}
\fi

\newsiamremark{remark}{Remark}
\newsiamremark{hypothesis}{Hypothesis}
\crefname{hypothesis}{Hypothesis}{Hypotheses}
\newsiamthm{claim}{Claim}

\headers{Competitive Bivirus SIS  Model with social distancing}{B. Catalano, K.~Paarporn and S.~Gracy}

\title{Analysis of a Competitive Bivirus SIS  Epidemic Model with Game theoretic social distancing\thanks{Submitted to the editors 
on \today. A  preliminary version of this paper has been accepted for  publication in the proceedings of the $2026$ American Control Conference \cite{ben:catalano:acc:26}.
\funding{Part of this material is based upon work supported by the National Science Foundation under Grants. 
 }}}

\author{Benjamin Catalano\thanks{Department of Electrical Engineering and Computer Science, South Dakota School of Mines and Technology, Rapid City, SD, USA.    (\email{benjamin.catalano@sdsmt.edu, sebin.gracy@sdsmt.edu}).  } 
\and Keith Paarporn\thanks{Department of Computer Science, University of Colorado- Colorado Springs, CO, USA.
  (\email{kpaarpor@uccs.edu})}
 \and Sebin Gracy\footnotemark[2]
 }

\usepackage{amsopn}
\DeclareMathOperator{\diag}{diag}

\ifpdf
\hypersetup{
  pdftitle={Game theoretic social distancing in Competitive Bi-virus SIS Epidemics},
  pdfauthor={B. Catalano, Keith Paarporn and S.~Gracy}
}
\fi

\makeatletter
\def\cref@override@label@type#1\@nil#2{#1}
\makeatother
\begin{document}
\maketitle

\begin{abstract}
We propose a competitive bi-virus model with dynamic social distancing behavior. Our model illustrates how public perception of different viruses changes the conditions for their eradication, their coexistence, or the dominance of one over the other. We show that our model is not monotone, in contrast to the classic bi-virus model. We detail how social distancing behavior produces different sets of equilibria than the classic bi-virus model and changes the criteria for their stability.
In particular, we detail the set of disease free equilibria (DFE) present in our model and identify necessary and sufficient conditions for almost global exponential stability of the same. We prove similar global results for all but one non-DFE isolated (unilateral) equilibria and local stability results for the remainder. We also consider coexistence equilibria; we show such equilibria, when they exist, take the form of lines of equilibria and give local conditions for their stability. Finally, we illustrate our theoretical findings with numerical examples.
\end{abstract}


\newenvironment{relevance}{\begin{@abssec}{Relevance to Life Sciences}}{\end{@abssec}}

\begin{relevance}
We focus on the spread of competing viruses (e.g., influenza and the common cold). Behavior which mitigates viral transmutability, e.g., social distancing, masking, etc. are widely 
discussed in the literature.
Our model, considers the two-virus setting and analyzes dynamic social behavior as a net response to the prevalence of both viruses. Our work rigorously illustrates how social behavior changes the effective infectivity of viruses and the conditions under which viruses become extinct or persist in the population.

\end{relevance}

\newenvironment{MathematicalContent}{\begin{@abssec}{Mathematical Content}}{\end{@abssec}}

\begin{MathematicalContent}
We use non-linear ordinary differential equations for modeling the spread of each of the two viruses. Individuals must dynamically decide whether or not to social distance.
We use replicator equations, a standard model in evolutionary game theory, to describe how these decisions change over time.
We construct a coupled system with disease and social distancing dynamics, and determine conditions for the existence and stability of various kinds of equilibria, including continua of equilibria 
for which we leverage Center Manifold Theory, where needed.

\end{MathematicalContent}



\section{Introduction}

\seb{Mathematical modeling of disease spread has been 
an active area of research since Bernoulli's work on modeling the spread of smallpox (Variola Major) in the $18^{th}$ century \cite{bernoulli1760essai}.
The practical significance of epidemic processes and  their downstream effects on several aspects of human life
has attracted the attention of multiple research communities, viz. physics \cite{newman}, economics \cite{bloom2018epidemics}, mathematics \cite{hethcote2000mathematics}, ecology \cite{munster2009avian}, and computer science \cite{prakash2010virus}.
A variety of modeling approaches exist in the literature; 
the basic idea underpinning almost all such approaches is the partitioning of the population into compartments based on stages of infection and recovery (e.g., susceptible, infected, exposed, recovering, etc.).
Various combinations of these compartments have spawn 
several models, resulting in a more nuanced understanding of the mechanisms that govern disease spread. In fact, the field of mathematical epidemiology has made tremendous progress over the last two centuries (see \cite{bailey1975mathematical}).
Still, applied forecasting and mitigation strategies have been limited when applied to real epidemics---see \cite{remuzzi2020covid} for early-days analysis of COVID-19.
We identify two major factors limiting the applicability of theoretical work: a) the emergence of new viral strains in the population during an endemic, and b) public perception of the danger of viruses and cost of pro-social policies.
}

\seb{Indeed, most of the existing disease spread models presume that there is only one virus (resp. strain of a virus) spreading in a population. However, it is not unusual for multiple viruses to simultaneously circulate in a population, 
for e.g., a) influenza and common cold \cite{Chan2018}, b) different RNA-viruses \cite{delgado2011evolution,piret2022viral,clarke1994red}. Additionally, multiple strains of a virus could also be in circulation at the same time \cite{nowak1991evolution,minayev2009improving,laurie2018evidence,wu2020interference}; some examples include strains of dengue fever \cite{pepin2008asymmetric} and those of coronavirus \cite{poland1996two}.  In such a context, the viruses involved could either be cooperative (see \cite{gracy2025modeling} and references therein for details), or 
competitive. The present paper focuses on the competitive case, i.e., suppose that there are two viruses, say virus 1 and virus 2, circulating in a community, then an individual  can either be infected with virus 1 or with virus 2 or with none.
Multi-virus systems (competitive or cooperative) exhibit far richer behavior than their single-virus counterparts \cite{newman2005threshold} - in particular,  multi-virus competitive models admit the possibility of coexistence (a scenario where both virus 1 and virus 2 exist in some sort of balance with each other) and competitive exclusion (a scenario where one virus pushes out the other). 
As such, various mathematical models have been developed to study multi-virus systems in depth, see \cite{sahneh2014competitive,zhang2025analysis,gracy2024towards,janson2024competitive,liu2019analysis,ye2021convergence,carlos2,karrer2011competing,ackleh2005competitive,wang2012dynamics,castillo1996competitive,gracy2026,ye2024competitive}. The present paper is based off of the competitive bi-virus susceptible-infected-susceptible (SIS) model. }

\seb{Although competitive networked epidemic models have been studied in-depth in recent years, most of these models suffer from a key limitation: 
these typically do not account for the public's reactions to a rapid outbreak. As seen during the COVID-19 crisis, the public's reaction to the imposition of social distancing measures was not necessarily uniform \cite{christensen2020political}. In fact, the population's willingness to comply with the mitigation measures such as social distancing could non-trivially impact the shape of a pandemic; for example, if a large fraction of the population chooses to not social distance, then that may result in a higher case load at the hospitals, thereby potentially overwhelming the healthcare facilities in a town or in a city. Therefore, there is a need for mathematical models that feature a co-evolution between decision-making (e.g. individual social distancing) and the spread of infectious diseases through physical contacts \cite{funk2015nine,heesterbeek2015modeling}. 
To better capture decision-making at the individual level, we adopt a game-theoretic formulation.}
In particular, evolutionary game-theoretic tools 
enable the modeling of dynamic social behavior in a variety of contexts \cite{sandholm2010population,weitz2016oscillating,quijano2017role,satapathi2025game}. 
In the context of behavior during epidemics, the perceived costs and benefits from taking or not taking social distancing actions are the basis of how individuals make such decisions. 
Importantly, these costs and benefits are linked to how widespread the disease currently is. 
\seb{Scalar SIS models coupled with  game-theoretic behavior have been proposed  \cite{paarporn2023sis,saad2023dynamics} -- however, these focus on the single-virus case. More recently, a single-node bi-virus SIS model, where susceptible individuals strategically choose whether to adopt protection or remain unprotected, to maximize their instantaneous
payoffs has been proposed in \cite{maitra2025bi}. However, the model in \cite{maitra2025bi} is limited in the sense that it does not consider the social distancing behavior of infected individuals
nor the dynamics associated with
social distancing. The present paper seeks to fill these gaps.}


Specifically, we examine a novel competitive bi-virus epidemic model in which social-distancing is modeled using game-theoretic dynamics.
We consider an individual's infection status as either susceptible or infected with one of the viruses. An individual chooses whether to social distance or not, where social distancing reduces contact rates with other individuals. 
We use replicator equations, a standard model in evolutionary game theory, to describe how these decisions change over time. \keith{In our game-theoretic formulation, susceptible individuals and infected individuals base their decisions on different incentive structures. Susceptible individuals base their decisions on costs associated with the perceived risks from being exposed to either one of the viruses, as well as the economic costs from social distancing (e.g. staying home). Our formulation accounts for infected individuals as either having more of an incentive to social distance (e.g. due to societal pressure), or having more of an incentive to disregard social distancing (e.g. due to selfish reasons). We describe this model 
with a system of five coupled ordinary differential equations.
\ben{
We phrase our results primarily in terms of the base and effective reproduction numbers for each virus. The base reproduction number of a virus is the expected number of additional infections per individual infected, in a completely susceptible population. The effective reproduction number is the base reproduction number mediated by the effective contact between infected and susceptible individuals; this takes into account partial immunity
(e.g., from social distancing, being vaccinated, already being infected, etc.).
Formal definitions of the base reproduction number and effective reproduction number are provided in equations~(\ref{eq:R0}-\ref{eq:R}), respectively. We refer to the virus with the larger effective reproduction number---given a particular population state---as the dominant virus.
}
}

\subsection*{Paper Contributions}
The main contribution of the present paper is to detail how social distancing behavior produces different sets of equilibria than the classic bi-virus model and, consequently,  changes the criteria for their stability. More specifically:
\begin{enumerate}[label=\roman*)]
\item We propose a novel bi-virus dynamical system with game-theoretic social distancing behavior; see equation~\eqref{eq:main}. 
\item We show that our model, in contrast to the classic bi-virus SIS dynamics, is \emph{not} a \seb{monotone dynamical system (MDS)}\footnote{A detailed discussion on MDS appears in Section~\ref{sec:monotone}.} 
(Theorem~\ref{thm:monotone}).
\item We classify all DFEs (Lemma~\ref{lem:all:DFEs}) of our model and provide necessary and sufficient conditions for local asymptotic stability of the same; namely that the effective reproduction 
number of each virus is less than unity at the DFE (Proposition~\ref{thm:DFE}). Moreover, we show that the aforementioned sufficient condition for local stability is also sufficient for almost global exponential stability of a given DFE (Theorem~\ref{thm:globalstability:DFE0}).
\item Thereafter, we restrict ourselves to the pro-social case (the payoff for social distancing is higher for both viruses than the cost of social distancing);
we identify sufficient conditions for almost-global asymptotic stability and instability of the unilateral equilibrium with no social distancing (Theorem~\ref{thm:fp:global}) (resp. with all social distancing (Proposition~\ref{prop:fp2:global})); namely that the effective reproduction number of the dominant virus at equilibrium is greater than unity, and the proportion of the population infected with the dominant virus is small enough (resp. large enough) such that the payoff for social distancing is negative (resp. positive) at the endemic equilibrium.
\item For non-generic parameter choices \footnote{The term  \enquote{generic} or \enquote{almost all} has a precise mathematical meaning: for all but a set of parameter values that has measure zero. This set of exceptional values is defined by an algebraic or semi-algebraic set.}  
(
specifically, for scenarios where the base reproduction numbers of each virus are identical), we identify conditions for existence and local stability of line segments of coexistence equilibria with no social distancing (Theorem~\ref{thm:l0}) (resp. with all social distancing (Theorem~\ref{thm:l1})); namely, the DFEs are all unstable, and the proportion of the population infected with either virus is small enough (resp. large enough) such that the payoff for social distancing is negative (resp. positive) on the line of equilibria. Similarly, we also provide a sufficient condition for the existence and local stability of a continuum of coexistence equilibria with intermediate social distancing (Theorem~\ref{thm:ls}). 
\end{enumerate}


\seb{The paper also makes the following auxiliary contributions: Assuming arbitrary costs (in the sense as stated in Section~\ref{sec:arbitrary:costs}), we identify conditions that guarantee only one class of equilibria may be stable, but all other classes of equilibria must be unstable; see Lemma~\ref{lem:arbit:equm:stability}. We also provide conditions  for the existence and local stability of an equilibrium, where virus~$1$ is endemic in the community and there is intermediate social distancing (a fraction of the population practices social distancing, and another fraction of the population does not); see Proposition~\ref{prop:p1s}. 
We also identify conditions for the existence of a continuum of equilibria with intermediate social distancing (see Lemma~\ref{lem:coexist}, item 3)). An overview of the present paper's contributions is provided in Table~\ref{table:FPs}.} 

\keith{We carry out simulation studies to verify our established theoretic results, and we also provide explorations of more complex dynamical behaviors that can be exhibited by our system by considering modified game-theoretic behaviors with regard to different types of payoff functions.} 

\ben{
Our model illustrates how the effective reproduction number of each virus is mediated by the proportion of the population which engages in social distancing and the effectiveness thereof. The effective reproduction number, being a function of the state, is implicitly dynamic, and each equilibrium admits a local effective reproduction number. Domains of stability can be understood, in part, as the domains about an equilibrium where the effective reproduction numbers admit convergence to said equilibrium.
Moreover, we find that, in our model, public perception of the threat of a virus may determine whether it remains endemic or dies out; furthermore, it is necessary to disambiguate the perception of \emph{each} virus to better understand multi-virus scenarios.
}

\seb{Some of the results in this paper have been accepted for publication  in the proceedings of the $2026$ American Control Conference \cite{ben:catalano:acc:26}. The differences between the present paper and \cite{ben:catalano:acc:26} are as follows:
\begin{enumerate}[label=\roman*)]
    \item Complete proofs of all results.
    \item Theorem~\ref{thm:monotone} is novel to this paper.
    \item None among Lemma~\ref{lem:all:DFEs}, Lemma~\ref{lem:arbit:equm:stability}, Proposition~\ref{thm:DFE}, and Theorem~\ref{thm:globalstability:DFE0} 
    appeared in \cite{ben:catalano:acc:26}.
    \item Neither Theorem~\ref{thm:fp:global} nor Proposition~\ref{prop:fp2:global} nor Theorem~\ref{thm:ls} were included in \cite{ben:catalano:acc:26}.
    \item An in-depth set of simulations, covering both the theoretical findings of this paper, and also exploring the effect of incorporating a payoff function that accounts for more complex behaviors.
\end{enumerate}
}

 \subsection*{Paper Outline}
 We conclude the present section by listing the notations  that will be used in the rest of the paper. We introduce our model in Section~\ref{sec:model}, while in Section~\ref{sec:monotone} we establish that it is not monotone. Section~\ref{sec:arbitrary:costs} explores the various equilibria (viz. existence and stability) that our model admits
 when the costs of taking social distancing measures and not taking social
distancing measures (i.e., the cost of putting other individuals at risk) are arbitrarily related. 
Section~\ref{sec:stability:DFE} provides a sufficient condition for almost global exponential stability (resp. instability) of the various DFEs.  Section~\ref{sec:analysis} identifies the various non-DFE isolated 
\seb{equilibria} of our model, Section~\ref{sec:unilateral:FPs}
secures conditions for almost global stability (resp. instability) of the various non-DFE 
\seb{equilibria}, and Section~\ref{sec:line} secures conditions for local asymptotic stability of different continua of coexistence equilibria. 
 We highlight our theoretical findings via numerical examples, and, additionally,  explore the effect of incorporating a co-ordination term into the dynamics which can be thought of as representing the possibility of social distancing being socially reinforced by those participating
 in Section~\ref{sec:sims}. Finally, we summarize our findings and present some problems of  interest to the wider community in Section~\ref{sec:conclusions} 

\vspace{-2mm}
\subsection*{Notation}
Let $\mathbb{R}$ and $\mathbb{R}_+$ denote the set of real numbers and the set of nonnegative real numbers, respectively.
We use $S\times U$ where $S$ and $U$ are sets to denote the Cartesian product of $S$ and $U$.
We denote logical conjunction and disjunction by $\land$ and $\lor$, respectively.
Given a positive integer $n$, we use $[n]$ to denote the set $\{1,2,\cdots,n\}$.
We use the Iverson bracket $[P]$ to denote the logical value of the predicate $P$, i.e., $[P] = 1$ if $P$ is true, otherwise, $[P] = 0$.
In the interest of conciseness, we omit the dimensions of vectors/matrices where they are unambiguous in context.
Given a vector $x$, the square matrix with the elements of $x$ along the diagonal is denoted by $\diag{(x)}$. Given a matrix $A$, $A_{ij}$ denotes the element in the $i^{th}$ row and the $j^{th}$ column of $A$.
We denote the $k^{th}$ row of the matrix $A$ by $A_{k,:}$.
Given two sets $A, B$, $A-B$ denotes the set difference operation.
For a set $S$, we denote the boundary of $S$ by $\partial S$ and the interior of $S$ by $\interior S$.
We denote the spectrum of matrix $A$ by $\spec(A)$, the spectral radius by $\rho(A)$, and the spectral abscissa by $s(A)$ (the spectral abscissa is defined as $s(A) := \max\{\rm{Re}(\lambda) \;|\; \lambda \in \spec(M)\}$).
For a state variable $x$, we write the fixed value of $x$ at some equilibrium as $x^*$.
We use the notation $O(\epsilon)$ to denote a term of order $\epsilon$; that is, $f(x) = O(\epsilon) \implies |f(x)| \le k\cdot|\epsilon|$ for some constant $k>0$.

\section{Model}\label{sec:model}
\seb{Consider a single well-mixed \footnote{Well-mixed means that the probability of any two individuals in a population interacting with each other is the same.} population of unit mass over which two viruses are spreading, denoted as virus~$1$ and virus~$2$. We assume that the aforementioned viruses are competitive. 
That is, each individual is either infected with virus~$1$, infected with virus~$2$, or healthy.
Specifically, no individual is infected with \textit{both} viruses simultaneously.
The term $s(t)$ denotes the fraction of individuals in the population that are susceptible (healthy) at time instant $t$; $y_i(t)$ for $i=1,2$ denotes the fraction of individuals in the population that are infected virus~$i$, and $y(t) := [y_1(t), y_2(t)]^\top$ denotes the vector of both masses. Note: $s(t) + y_1(t) + y_2(t) = 1$. A healthy individual is infected with virus~$i$ based on the infection rate $\beta_i$; an individual infected with virus~$i$ recovers based on the healing rate $\delta_i$.
The classic SIS bi-virus model is a mean-field approximation of a $3$-state Markov process (see \cite{liu2019analysis} for more details), where the terms $\beta_i$ (resp. $\delta_i$) can be thought of as the average instantaneous rate of infection (resp. rate of recovery) 
with respect to virus~$i$; in isolation $\beta_i$ represents the chance of infection per unit mass of contact between healthy and infected individuals, whereas $\delta_i$ represents the exponential decay rate of the infected population (again in isolation).

We have the following assumption.}


\begin{assumption}\label{assume:betadelta}
    The spreading rates satisfy $\beta_1,\beta_2 > 0$, and the healing rates satisfy $\delta_1,\delta_2 > 0$.
\end{assumption}


Individuals choose to either follow social distancing or not (action $\mcd$ or $\mcn$, resp.). We let $x_{\S\mcd}(t) \in [0,s(t)]$ and $x_{\S\mcn}(t) = s(t) - x_{\S\mcd}(t)$ denote the mass of susceptible individuals that follow and do not follow social distancing, respectively. Likewise, we denote $x_{i\mcd}(t) \in [0,y_i(t)]$ and $x_{i\mcn}(t) = y_i(t) - x_{i\mcd}(t)$ as the mass of individuals infected with virus $i=1,2$ that follow and do not follow social distancing, respectively.
The virus infections evolve according to the following dynamics:
\begin{equation}
    \begin{aligned}
        \dot{y}_1(t) &= \beta_1(qx_{\S\mcd}(t) + x_{\S\mcn}(t))(qx_{1\mcd}(t) + x_{1\mcn}(t)) - \delta_1 y_1(t) \\
        \dot{y}_2(t) &= \beta_2(qx_{\S\mcd}(t) + x_{\S\mcn}(t))(qx_{2\mcd}(t) + x_{2\mcn}(t)) - \delta_2 y_2(t) ,\\
    \end{aligned}
\end{equation}

\seb{where $q$ denotes the factor by which interaction among individuals reduces as a consequence of social distancing. We have the following assumption.} 

\begin{assumption}\label{assume:q}
   \sebcancel{The factor $q$ satisfies the following}:  $q \in (0,1)$.
\end{assumption}
When $q$ is close to zero, the social distancing action provides more protection (e.g. quarantining). When $q$ is higher, the social distancing action provides less protection (e.g. only wearing a mask but still frequenting public places). 

\vspace{-2mm}
\subsection{Payoff functions}
\seb{Social distancing has both costs (e.g., increased social isolation, reduced economic activity \cite{ye2021game}) and benefits (e.g.,reduced chances of falling sick). Therefore, individuals need to make a decision on whether or not to social distance. Such  a decision making process can be modeled using payoff functions.}
For susceptible individuals, the perceived payoff for choosing to social distance is given by
\begin{equation}
    \pi_{\S\mcd}(y(t)) := -c_\mcd + r_1 y_1(t) + r_2 y_2(t)
\end{equation}
where $c_\mcd > 0$ is the economic and social cost of taking social distancing measures. Individuals do not know the true probability of getting infected, but are typically informed about the total amount of people currently infected. Thus, we have defined the parameters $r_i > 0$ as perceived risk factors to being exposed to virus $i = 1,2$. Consequently, the perceived payoff for social distancing is increasing in the mass of infected individuals of either virus type.  \seb{We have the following assumption.}

\begin{assumption}\label{assume:r}
    The terms $r_1, r_2$ must satisfy $0 < r_1 < r_2$.
\end{assumption}

This assumption asserts that virus 1 is  perceived among the population to not be as severe as virus 2. The perceived payoff for choosing to not social distance is

\begin{equation}\label{eq:payoff:sn}
    \pi_{\S\mcn}(y(t)) := - (r_1 y_1(t) + r_2 y_2(t)).
\end{equation}
By not social distancing, a susceptible individual does not pay the cost $c_\mcd$, but pays a cost based on the perceived risks. 
For infected individuals of virus $i=1,2$, 
the perceived payoffs for choosing social distancing or not \seb{can be modeled} as follows:
\begin{equation}
        \pi_{i\mcd}(y(t)) := -c_\mcd \quad \text{and} \quad
        \pi_{i\mcn}(y(t)) := -c_i
\end{equation}
\seb{Observe that} infected individuals that social distance pay the cost $c_\mcd > 0$ \seb{as mentioned previously}, while those that do not social distance pay a perceived cost $c_i > 0$, \seb{in view of the fact that they are} putting other individuals at risk to virus $i$. In this initial formulation, we will assume that $c_i > c_\mcd$, but our forthcoming analysis includes results where the opposite is true.

%
\normalsize

\subsection{Coupled evolutionary dynamics}
\seb{A very powerful tool for describing the evolution of social distancing behaviors is the notion of replicator equations; see \cite{schuster1983replicator} for an in-depth 
view of replicator equations. Our choice of replicator equations, is motivated by the fact that they
arise in evolutionary game theoretic models from many distinct types of revision protocols \cite{sandholm2010population}, i.e., individual-level decision-making processes. Thus, they accommodate a range of behaviors. One such protocol is a pairwise imitation protocol, where agents adopt the strategies of other agents who obtain high payoffs. 
Specifically in the context of epidemic spreading, it is not necessary for 
individuals to physically interact with one another in order to learn about 
payoffs of others -- these may be learned through interactions with media and social networks. Thus, we assume the interaction that underlies the replicator equation does not factor in to physical epidemic spreading process.
}

Define $z_\S(t) := x_{\S\mcd}(t)/s(t)$ as the \emph{fraction} of individuals among the susceptible population that social distance. Likewise, define $z_1(t) := x_{1\mcd}(t)/y_1(t)$ and $z_2(t) := x_{2\mcd}(t)/y_2(t)$. 
Observe that, for a given $t \in \mathbb{R}_{\geq 0}$, $z_1(t)$ (resp. $z_2(t)$) are defined only if $y_1(t)\neq 0$ (resp. $y_2(t) \neq 0$). The following replicator equation describes the evolution of social distancing behaviors among the susceptible population.
\begin{equation}
    \dot{z}_\S(t) = z_\S(t)(1-z_\S(t))(\pi_{\S\mcd}(y(t)) - \pi_{\S\mcn}(y(t)))
\end{equation}


Likewise, we also have replicator equations for the two infected subpopulations,
\begin{equation}
    \begin{aligned}
        \dot{z}_1(t) &= z_1(t)(1-z_1(t))(\pi_{1\mcd}(y(t)) - \pi_{1\mcn}(y(t))) \\
        \dot{z}_2(t) &= z_2(t)(1-z_2(t))(\pi_{2\mcd}(y(t)) - \pi_{2\mcn}(y(t)))
    \end{aligned}
\end{equation}

Therefore, with the dependency on the time variable $t$ being omitted, the full set of coupled equations is:
\begin{equation}\label{eq:main}
    \begin{aligned}
        \dot{y}_1 &= y_1 (\beta_1 s (1 - z_\S(1-q))(1 - z_1(1-q)) - \delta_1) \\
        \dot{y}_2 &= y_2 (\beta_2 s (1 - z_\S(1-q))(1 - z_2(1-q)) - \delta_2) \\
        \dot{z}_\S &= z_\S(1-z_\S)(2(r_1y_1 + r_2y_2) - c_\mcd) \\
        \dot{z}_1 &= z_1(1-z_1)(c_1 - c_\mcd) \\
        \dot{z}_2 &= z_2(1-z_2)(c_2 - c_\mcd) 
    \end{aligned}
\end{equation}

Here, we have used the fact that each of the states $x_{\S\mcd}$, $x_{1\mcd}$, and $x_{2\mcd}$ determine the other states $x_{\S\mcn} = s(t) - x_{\S\mcd}$, $x_{1\mcn} = y_1 - x_{1\mcd}$, and $x_{2\mcn} = y_2 - x_{2\mcd}$.




We classify 
\seb{equilibria} into the following three categories.
\begin{itemize}
    \item A \emph{disease-free equilibrium} (DFE) is any \seb{equilibrium} for which $y_1 = y_2 = 0$.
    \item A \emph{unilateral equilibrium} is any \seb{equilibrium} for which either $y_1 = 0$ and $y_2 > 0$, or $y_1 > 0$ and $y_2 = 0$.
    \item A \emph{coexistence equilibrium} is any \seb{equilibrium} for which $y_1,y_2 > 0$.
\end{itemize}

We define the sets $\Delta$ and $\Gamma$ as follows:
\begin{align}
\Delta &:= \left\{(y_1,y_2) \;\big|\; \sum y_i \le 1 \land 0 \le y_i, i\in[2]\right\}\label{eq:Delta} \\
\Gamma &:= \Delta\times[0,1]^3 \label{eq:gamma}
\end{align}
We define the full state vector as $x(t):=[
    y_1(t) ,y_2(t) ,z_s(t) , z_1(t), z_2(t)
]^\top$.
We need the following assumption to ensure that our model is well-defined.
\begin{assumption}\label{assume:init}
The initial states satisfy
$x(0)\in \Gamma$.   
\end{assumption}
With Assumption~\ref{assume:init} in place, the following Lemma guarantees that the set $\Gamma$ is positively invariant.
\begin{lemma}\label{lem:pos_inv}
  Consider system~\eqref{eq:main} under Assumption~\ref{assume:betadelta}. The set $\Gamma$, where $\Gamma$ is as defined in~\eqref{eq:gamma}, is positively invariant. 
\end{lemma}
\textit{Proof:} Observe that $\Gamma$ is a closed set, so $\partial\Gamma\subset\Gamma$ and it is impossible for a continuous trajectory to leave $\Gamma$ without passing through $\partial\Gamma$.
Therefore, we examine the behavior of system~\eqref{eq:main} at $\partial\Gamma$.
Consider $y_1 = 0$ (resp. $y_2 = 0$), then $\dot y_1 = 0$ (resp. $\dot y_2 = 0$). If $y_1 + y_2 = 1$, then $\dot y_i = -\delta_i y_i$, for $i\in[2]$.  
Due to Assumption~\ref{assume:betadelta}, $\delta_i > 0, i\in[2]$. Therefore, it follows that at the boundary, $\partial\Delta$, $y_1$ (resp. $y_2$) is either at the lower bound, $0$, and constant or at the upper bound, $y_1 + y_2 = 1$, and decreasing.\\
Consider $z\in\{z_\S, z_1, z_2\}$. In each case, $\dot z$ 
depends on
$z(1-z)$. Since $z\in[0,1]$, the boundary values are given by $\partial[0,1]=\{0,1\}$, thus, evaluating the right hand side of the last three lines of~\eqref{eq:main} at each of the values in $\partial[0,1]$, we get $\dot z = 0$. Hence, $z_\S$, $z_1$ and $z_2$ are always constant 
at their respective boundaries.\\
Therefore, defining $x(t^\prime)=[\begin{smallmatrix}
    y_1(t^\prime) & y_2(t^\prime) &z_s(t^\prime) & z_1(t^\prime) &z_2(t^\prime)
\end{smallmatrix}]^\top$, it must be that if $x(t_0)\in\Gamma$ then $x(t)\in\Gamma, \forall t > t_0$; that is, $\Gamma$ is positively invariant.~$\square$
\\
Lemma~\ref{lem:pos_inv} guarantees that each state variable takes values in the interval $[0,1]$, and ensures that the model is physically interpretable in terms of population proportions.

For ease of notation, define the expression $q_{(\cdot)} := 1 - (\cdot)(1 - q)$. Observe that any $z\in[0,1]$, we have $q_z\in(q,1)$.
Then the fully generalized Jacobian of the system, $J(x)$, evaluated at an arbitrary point $x\in\Gamma$ is given by Equation~\eqref{eq:jacobian:gen}.
\begin{equation}\label{eq:jacobian:gen}\scriptsize
\begin{bmatrix}
\beta_1 (s - y_1) q_{z_S} q{z_1} - \delta_1 &
-\beta_1 y_1 q_{z_S} q_{z_1} &
-\beta_1 y_1 s (1 - q) q_{z_1} &
-\beta_1 y_1 s q_{z_S} (1 - q) &
0 \\
-\beta_2 y_2 q_{z_S} q_{z_2} &
\beta_2 (s - y_2) q_{z_S} q_{z_2} - \delta_2 &
-\beta_2 y_2 s (1 - q) q_{z_2} &
0 &
-\beta_2 y_2 s q_{z_S} (1 - q) \\
z_S(1 - z_S) 2 r_1 & z_S(1 - z_S) 2 r_2 &
(1 - 2 z_S)(2\textbf{r}^\top\textbf{y} - c_\mathcal{D}) &
0 & 0 \\
0 & 0 & 0 & (1 - 2 z_1)(c_1 - c_\mathcal{D}) & 0 \\
0 & 0 & 0 & 0 & (1 - 2 z_2)(c_2 - c_\mathcal{D}) \\
\end{bmatrix}
\end{equation}

\ben{
We define the following quantities for each virus, $i\in[2]$:
\begin{align}
    R_i^0 &:= \beta_i / \delta_i
    &\text{base reproduction number}\label{eq:R0}\\
    R_i &:= q_{z_\S} q_{z_i} \beta_i / \delta_i
    &\text{effective reproduction number}\label{eq:R}\\
    R_i^* &:= q_{z_\S^*} q_{z_i^*} \beta_i / \delta_i
    &\text{equilibrium reproduction number}\label{eq:R_eq}
\end{align}
\begin{remark}[Interpretation of Reproduction Numbers]
For each virus $i$, the quantity $R_i^0$ is the classical base reproduction number; $R_i$ is the effective effective reproduction number given the current population state; and $R_i^*$ is the effective reproduction number at the implicit equilibrium, $x^* = [y_1^*, y_2^*, z_\S^*, z_1^*, z_2^*]^\top$.
\end{remark}
Note that $R_i^0$ and $R_i^*$ (in the context of a fixed equilibrium point, $x^*$) are constants, while $R_i$ is variable throughout a trajectory. Observe also that $R_i^*$ is a function of $z_\S^*$ and $z_i^*$ but not directly $y^*$. 
 }


\begin{table*}\footnotesize
    \begin{center}
        \begin{tabular}{|c|c|c|}
        \hline
        Equilibrium  & Existence condition & Stability condition \\
        \hline
        $\bs{p}_{\DFE0} = [0,0,0,1,1]^\top$ & Always & Proposition~\ref{thm:DFE} \\        
        \hline
        $\bs{p}_{\DFE1} = [0,0,1,1,1]^\top$ & Always & Proposition~\ref{thm:DFE} \\        
        \hline
        $\bs{p}_{10} = [1 - (q R_1^0)^{-1},0,0,1,1]^\top$  & $R_1^* > 1$ & Theorem~\ref{thm:fp:global} \\        
        \hline
        $\bs{p}_{11} = [1 - (q^2 R_1^0)^{-1},0,1,1,1]^\top$ & $R_1^* > 1$ & Prop.~\ref{prop:fp2:global} \\        
        \hline
        $\bs{p}_{1\S}$ \eqref{lem:unilateral}
        & $\frac{c_\mcd}{2r_i} < 1$, $q (1-\frac{c_\mcd}{2r_i})^{-1} < q^2 R_i^0 < (1-\frac{c_\mcd}{2r_i})^{-1}$ & Prop.~\ref{prop:p1s} \\        
        \hline
        $\bs{p}_{20}$ analogous to $\bs{p}_{10}$ & $R_2^* > 1$ & Analogous to Theorem~\ref{thm:fp:global} \\        
        \hline
        $\bs{p}_{21}$ analogous to $\bs{p}_{11}$ & $R_2^* > 1$ & Analogous to Prop.~\ref{prop:fp2:global} \\        
        \hline
        $\bs{p}_{2\S}$ analogous to $\bs{p}_{1\S}$ & $\frac{c_\mcd}{2r_i} < 1$, $q (1-\frac{c_\mcd}{2r_i})^{-1} < q^2 R_i^0 < (1-\frac{c_\mcd}{2r_i})^{-1}$ & Analogous to Prop.~\ref{prop:p1s} \\
        \hline
        Line $\mcal{L}_0$ \eqref{eq:line0} & $R_1^0 = R_2^0$, $q R^0 > 1$ & Theorem~\ref{thm:l0}  \\
        \hline
        Line $\mcal{L}_1$ \eqref{eq:line1} & $R_1^0 = R_2^0$, $q^2 R^0 > 1$ & Theorem~\ref{thm:l1} \\
        \hline
        Line $\mcal{L}_\S$ \eqref{eq:lineS} & Lemma~\ref{lem:coexist} (3)  &
        Theorem~\ref{thm:ls}        
        \\
        \hline
        \end{tabular}
    \end{center}
    \caption{Table of fixed points, $[
    y_1 ,y_2 ,z_s , z_1, z_2
]^\top$ and their sufficient conditions for existence and their stability properties. \seb{The DFE analysis appears in Section~\ref{sec:arbitrary:costs}, the analysis of the unilateral \seb{equilibria} 
    appears in Section~\ref{sec:unilateral:FPs}, while that of the continua of coexistence equilibria is provided in Section~\ref{sec:line}.}}
    \label{table:FPs}
\end{table*}



\section{Non-monotonicity 
of system~\eqref{eq:main}}\label{sec:monotone}


\seb{\emph{Monotone dynamical systems} (MDS) form a particular class of nonlinear systems.
An $n$-dimensional nonlinear system, $\dot{x} = f(x)$, is said to be $K_m$ monotone, if there is a uniform partial ordering of trajectories in the system \cite{ye2021convergence}. That is,
$$x(0) \leq_{K_m} y(0) \implies x(t) \leq_{K_m} y(t), \forall t \in \mathbb{R}_{\ge0}$$
Where $m\in\{0,1\}^n$, and the relation $\le_{K_m}$ is defined by:
$$x(t) \le_{K_m} y(t) \iff (-1)^{m_i} \cdot x_i(t) \le (-1)^{m_i} \cdot y_i(t), i\in[n]$$
This property that MDS share has been heavily leveraged to draw conclusions regarding the \emph{typical} behavior of certain classes of nonlinear systems; see \cite{hirsch1985systems,smith1988systems}. As Morris Hirsch shows in his seminal paper \cite{hirsch1985systems}, if a nonlinear system is MDS, then for almost all initial conditions, the dynamics converge to an equilibrium point. In the context of the present paper, the importance of the notion of MDS is to be understood as follows: 
supposing that  system~\eqref{eq:main} is monotone, then the typical behavior of the dynamics of system~\eqref{eq:main} is  convergence to some equilibrium point (disease-free, endemic, coexistence, etc.). More pertinently, the  existence of limit cycles (i.e., the occurrence of waves of epidemic) is less likely. 
Furthermore,  limit cycles (if any)  would be non-attractive \cite{ye2021convergence}.  System~\eqref{eq:main} does not admit any more complicated behavior \cite{sontag2007monotone}. On the contrary, if the system is not monotone, then  no dynamical behavior, including chaos, can be definitively ruled out without additional analysis \cite{sontag2007monotone}. Therefore, it is important to determine whether or not system~\eqref{eq:main} is monotone. Our first main result is the following.}

\begin{theorem}\label{thm:monotone}
    Under Assumptions~\ref{assume:betadelta}-\ref{assume:r}, system~\eqref{eq:main} is not monotone.
\end{theorem}
\textit{Proof:} Let the vector $x(t)$ be as defined in Assumption~\ref{assume:init}. Let $J(x)$ denote the Jacobian computed at an arbitrary point $x(t) \in \Gamma$.
From \cite[Lemma~2.1]{smith1988systems} and \cite[page 424]{hirsch1985systems}, system~\eqref{eq:main} is monotone if, and only if, there exists a diagonal matrix $P$ s.t. $P_{ii}\in\{-1,1\}, \forall i\in[n]$ whereby $P J(x) P$ is Metzler $\forall x\in \Gamma$. We show that no such coordinate transform exists which makes the Jacobian, $J(\cdot)$, Metzler over $\Gamma$.
We proceed as follows.

\noindent Let $P$ be an arbitrary matrix as defined above.  
Note that $[P J(x) P]_{ij} = J(x)_{ij} P_{ii} P_{jj}$ and $[P J(x) P]_{ji} = J(x)_{ji} P_{ii} P_{jj}$, and since $0\notin \diag(P)$,  $\sgn([P J(x) P]_{ij}) = 0$ if, and only if, $J(x)_{ij} = 0$.
Thus, $P J(x) P$ cannot be Metzler, if $\exists i,j\in[n]$ s.t. $J(x)_{ij} < 0 < J(x)_{ji}$; otherwise, either $J(x)_{ij} P_{ii} P_{jj}$ or $J(x)_{ij} P_{ii} P_{jj}$ will be negative under any $P$, and $P J(x) P$ cannot be Metzler.
Consider $x\in\interior\Gamma$. Observe that under Assumptions~\ref{assume:betadelta}-\ref{assume:r}, $J(x)_{13} < 0$ and $J(x)_{31} > 0$. Thus, $P J(x) P$ cannot be Metzler for any $P$, and from \cite[Lemma~2.1]{smith1988systems} and \cite[page 424]{hirsch1985systems}, system~\eqref{eq:main} is not monotone.~\qed

Theorem~\ref{thm:monotone} highlights an interesting phenomenon - namely, that system~\eqref{eq:main} is not a monotone dynamical system (MDS). This is in sharp contrast to  classical bi-virus systems (i.e., $q=1$ in system~\eqref{eq:main}), which are indeed  MDS; see \cite{ye2021convergence}. 
Given that system~\eqref{eq:main} is not monotone, we cannot use the rich theory of MDS to draw overarching conclusions regarding the typical (resp. limiting) behavior of our model. In fact,   in view of Theorem~\ref{thm:monotone}, as mentioned previously, no dynamical behavior, including chaos, can be definitively ruled out. That said, one approach for studying the limiting behavior of system~\eqref{eq:main} is the theory of singular perturbations \cite{wang2006remark};  a detailed investigation of this issue via 
the aforementioned approach is, however, beyond the scope of the present paper.

\section{The case where the costs $c_i$, $c_\mcd$ are arbitrary}\label{sec:arbitrary:costs} 
In this section, we explore the scenario where we admit costs $c_i$, for $ i \in [2]$, and $c_{\mcd}$ to be arbitrary in the sense that we do not place any restrictions on how the costs $c_i$ for $i \in [2]$ are related to $c_\mcd$.
The following lemma describes exhaustively the set of all possible DFEs that system~\eqref{eq:main} can posses.
\begin{lemma}\label{lem:all:DFEs}
    Consider system~\eqref{eq:main}. For generic choices of $c_\mcd, c_i, i\in[2]$, the system has eight DFEs given by $x^*\in\{0\}^2\times\{0,1\}^3$.
\end{lemma}
\textit{Proof}: By definition, at a DFE, $y_1 = y_2 = 0$. Then, from \eqref{eq:main}, it is immediate that $\dot{y}_1 = \dot{y}_2 = 0$. Substituting $y_1 = y_2 = 0$ into $\dot{z}_\S$ gives $\dot{z}_\S = 0 (= -c_\mcd z_\S (1 - z_\S))$, 
whose roots are $0$ and  $1$. Similarly, for generic choice of $c_i, c_\mcd$, we have $c_i - c_\mcd \ne 0$, thus, $\dot{z}_i = 0 (= z_1 (1 - z_1) (c_i - c_\mcd))$ has the roots $0$ and $1$.
Combining the solutions yields the complete set of DFEs: $\{0\}^2\times\{0,1\}^3$.~\qed

\begin{lemma}\label{lem:arbit:equm:stability}
    Consider system~\eqref{eq:main} with generic choices of $c_\mcd, c_i$, $\forall i\in[2]$. At any equilibrium, $z_i^*\in\{0,1\}$, $\forall i\in[2]$.
    Furthermore, the set of equilibria can be partitioned into four sets, corresponding to each choice of $(z_1^*, z_2^*)\in\{0,1\}^2$, where each partition is equivalent modulo the values of $z_1^*$ and $z_2^*$. Only the partition corresponding to $(z_1^*, z_2^*) = ([c_1 > c_\mcd], [c_1 > c_\mcd])$ may contain stable equilibria.
\end{lemma}
\textit{Proof:} Consider the Jacobian $J(x)$, given in~\eqref{eq:jacobian:gen}. Note that $J(x)$ can be written as follows:
$J(x) = [\begin{smallmatrix}
A & B \\
\textbf{0}_{2\times3} & C \\
\end{smallmatrix}]$.
Further, it suffices for us to consider the eigenvalues of matrix $C$ to establish our result, since $J(x)$ is block upper triangular and matrix $A$ does not play a role in the analysis. Hence, we will focus on the same.
Therefore, the eigenvalues of $J(x)$ are given by the eigenvalues of $A$ and $C$, where 
Observe that $C$ is as follows:
\begin{align*}
C &= \begin{bmatrix}
(1 - 2 z_1)(c_1 - c_\mathcal{D}) & 0 \\
0 & (1 - 2 z_2)(c_2 - c_\mathcal{D}) \\
\end{bmatrix}
\end{align*}
Immediately, since $C$ is diagonal, we have:
$$\spec(C) = \{(1 - 2 z_1)(c_1 - c_\mathcal{D}), (1 - 2 z_2)(c_2 - c_\mathcal{D})\}$$
Consider $\dot z_i, i\in[2]$.
Solving for $\dot z_i = 0$, we 
see that $C$ can be Hurwitz if, and only if, $(z_i\in\{0,1\}) \lor (c_i = c_\mathcal{D})$.
As in Lemma~\ref{lem:all:DFEs}, for generic choice of $c_\mcd, c_i, \forall i\in[2]$, we have $c_i - c_\mcd \ne 0$, so at \textit{any} equilibrium point, $z_i^*\in\{0,1\}, \forall i\in[2]$. Hence, for each $ i \in [2]$, the following must be true:
$$\begin{cases}
z_i^* = 0 \implies c_i - c_\mathcal{D} \in \spec(J(x^*)) \\
z_i^* = 1 \implies c_\mathcal{D} - c_i \in \spec(J(x^*))
\end{cases}$$
From this note two things: First, for generic choices of $c_\mathcal{D}, c_i, i\in[2]$, the matrix $C$ is never degenerate at an equilibrium point, and so application of the Hartman-Grobman Theory is never prevented by the eigenvalues of matrix $C$. 
Second, the set of equilibria for which $z_i^* = 0$ and for which $z_i^* = 1$ cannot be simultaneously stable.
Then, concisely,
\begin{align*}
\bigwedge_{i\in[2]}
\left(\left(1 - 2z_i^*\right) (c_i - c_\mathcal{D}) < 0 \right)
&\iff s(C) < 0 \implies (x^* \text{ is stable} \iff s(A) < 0)\\
\bigvee_{i\in[2]}
\left(\left(1 - 2z_i^*\right) (c_i - c_\mathcal{D}) > 0 \right)
&\iff s(C) > 0 \implies x^* \text{ is unstable}
\end{align*}
This gives exhaustive conditions for $C$ to be Hurwitz.
It follows that if $c_i > c_\mcd$ then all stable equilibria must have $z_i^* = 1$, and all equilibria with $z_i^* = 0$ are unstable, and vice versa for $c_i < c_\mcd$. Thus, the equilibria may be partitioned into four subsets based on the values of $z_1^*$ and $z_2^*$. Moreover, only the subset with $z_i^* = [c_i > c_\mcd], \forall i\in[2]$ may contain stable equilibria. $\square$

With Lemma~\ref{lem:arbit:equm:stability}, $R_i^*$ is fully determined by the model parameters up to the factor $q_{z_\S^*}\in[q,1]$.

While Lemma~\ref{lem:all:DFEs} details the DFEs of system~\eqref{eq:main}, we have not addressed when a DFE may be locally (or globally) stable.
The following theorem addresses local stability of the DFEs; see Theorem~\ref{thm:globalstability:DFE0} for domains of attraction.

\begin{proposition}\label{thm:DFE}
    Consider system~\eqref{eq:main} under Assumptions~\ref{assume:betadelta} and \ref{assume:q}.
    For generic choices of $c_i$ ($ i \in [2]$) and $c_\mcd$, there is at most one stable DFE, as given below:
    \begin{enumerate}[label=\roman*)]
        \item If $c_i < c_\mcd$ and $\beta_i/\delta_i < 1, i\in[2]$, $x^* = [0, 0, 0, 0, 0]^\top$ is the unique stable DFE. \label{thm:DFE:i_start}
        \item If $c_1 < c_\mcd, c_2 > c_\mcd$ and $\beta_1/\delta_1 < 1, q\beta_2/\delta_2 < 1$, $x^* = [0, 0, 0, 0, 1]^\top$ is the unique stable DFE.
        \item If $c_1 > c_\mcd, c_2 < c_\mcd$ and $q\beta_1/\delta_1 < 1, \beta_2/\delta_2 < 1$, $x^* = [0, 0, 0, 1, 0]^\top$ is the unique stable DFE.
        \item If $c_i > c_\mcd$ and $q\beta_i/\delta_i < 1, i\in[2]$, $x^* = [0, 0, 0, 1, 1]^\top$ is the unique stable DFE. \label{thm:DFE:i_end}
        \item If none of the conditions in \ref{thm:DFE:i_start}-\ref{thm:DFE:i_end} are satisfied, then
        there 
        exists 
        no stable DFE.
    \end{enumerate}
\end{proposition}
\textit{Proof:}
Consider the Jacobian (as given in~\eqref{eq:jacobian:gen}), evaluated at an arbitrary DFE $x^* = [0, 0, z_\S^*, z_1^*, z_2^*]$. This yields:
\begin{equation}\scriptsize
J(x^*) = \begin{bmatrix}
\beta_1 q_{z_\S^*} q_{z_1^*} - \delta_1 & 0 & 0 & 0 & 0 \\
0 & \beta_2 q_{z_\S^*} q_{z_2^*} - \delta_2 & 0 & 0 & 0 \\
z_S^*(1 - z_S^*) 2 r_1 &
z_S^*(1 - z_S^*) 2 r_2 &
-c_\mathcal{D} (1 - 2 z_S^*) &
0 & 0 \\
0 & 0 & 0 & (1 - 2 z_1^*)(c_1 - c_\mathcal{D}) & 0 \\
0 & 0 & 0 & 0 & (1 - 2 z_2^*)(c_2 - c_\mathcal{D}) \\
\end{bmatrix}
\end{equation}
Observe that the Jacobian is lower triangular, so its spectrum is immediately its diagonal entries. 

At equilibria, $\dot{z}_\S = 0$ implies $z_\S^*\in\{0,1\}$. From the third row of the Jacobian, we see that, since $c_{\mcd} > 0$, stability requires $z_\S^* < 1/2$; thus, we conclude that $z_\S^* = 0$ at any stable DFE. By similar examination, we find the values of $z_i^*$ based on the constants $c_\mcd, c_i, i\in[2]$; that is, using Iverson bracket notation,  $z_i^* = [c_i > c_\mcd], i\in[2]$. Then, the stable DFE candidate is $x^* = [0, 0, 0, [c_1 > c_\mcd], [c_2 > c_\mcd]]^\top$; all other DFEs are unstable. With this, we simplify the Jacobian $J(x^*)$ further:
\begin{equation}\scriptsize
J(x^*)=\begin{bmatrix}
\beta_1 q_{[c_1 > c_\mcd]} - \delta_1 & 0 & 0 & 0 & 0 \\
0 & \beta_2 q_{[c_2 > c_\mcd]} - \delta_2 & 0 & 0 & 0 \\
0 & 0 & -c_\mathcal{D} & 0 & 0 \\
0 & 0 & 0 & (1 - 2 [c_1 > c_\mcd])(c_1 - c_\mathcal{D}) & 0 \\
0 & 0 & 0 & 0 & (1 - 2 [c_2 > c_\mcd])(c_2 - c_\mathcal{D}) \\,
\end{bmatrix}
\end{equation}
where 
$q_{[c_i > c_\mcd]}$ is $q$ if $c_i > c_\mcd$ and $1$ otherwise.
Therefore, if $c_i < c_\mcd$, stability requires $\beta_i/\delta_i < 1$; otherwise, if  $c_i > c_\mcd$, then stability requires $q\beta_i/\delta_i < 1$ for each $i\in[2]$.
Therefore, in accordance with Lemma~\ref{lem:arbit:equm:stability}, for a given generic choice of $c_i, c_\mcd$, there is at most one stable DFE, which is given by $x^* = [0, 0, 0, [c_1 > c_\mcd], [c_2 > c_\mcd]]^\top$, and which is stable if $q_{[c_i > c_\mcd]}\frac{\beta_i}{\delta_i} < 1, i\in[2]$.

\begin{remark}\label{rem:qR}[Variation in effective reproduction numbers based on costs]
Observe that the conditions for the stability of the DFE in Proposition~\ref{thm:DFE} are almost identical for all regimes of $c_i, c_\mcd$, varying only by a factor of $q$. Generally, we find that the traditional conditions, $\beta_i/\delta_i < 1$, are mediated by a factor of $q_{z_\S}q_{z_i}$, which typically translates to $q^k$ for some $k\in\{0, 1, 2\}$.
\end{remark}

\section{Stability analysis of DFE}\label{sec:stability:DFE}
In this section, we secure a condition for almost global exponential stability (resp. instability) of each DFE. 
First, we define the following domain:
\begin{equation}\label{eq:D}
\mathcal{D} := \Delta\times[0,1)\times((0,1)\cup\{[c_1 > c_\mcd]\})\times((0,1)\cup\{[c_2 > c_\mcd]\})\supset\interior\Gamma
\end{equation}
We have the following result.
\begin{theorem}\label{thm:globalstability:DFE0}
   Consider system~\eqref{eq:main} under Assumptions~\ref{assume:betadelta} and \ref{assume:q} with generic choices of $c_\mcd, c_i, i\in[2]$.
   \begin{enumerate}[label=\roman*)]
       \item \label{q1}
       If for some $i \in [2]$, $R_i^* = q_{[c_i > c_\mcd]}\beta_i/\delta_i \ge 1$, then there are no stable DFEs.
       \item \label{q2} If for all $i\in[2]$, $R_i^* = q_{[c_i > c_\mcd]}\beta_i/\delta_i < 1$, then $x^* = [0, 0, 0, [c_1 > c_\mcd], [c_2 > c_\mcd]]^\top$ is the only stable DFE and is exponentially stable with the  domain of attraction $\mathcal{D}$ (as defined in~\eqref{eq:D}).
   \end{enumerate}  
\end{theorem}
\textit{Proof:} 
\textit{Proof of statement~\ref{q1}}: 
First observe that for generic but fixed choices of $c_\mcd, c_i, i\in[2]$, the point $x^*$ must be one of the four DFEs specified in Proposition~\ref{thm:DFE}, statements~(i)-(iv). Depending on what value $[c_i > c_\mcd]$ takes and if $q_{[c_i > C_\mcd]}\cdot\beta_i/\delta_i < 1, i\in[2]$, then the corresponding condition (i.e., the conditions in statement (i) or in (ii) or in (iii) or in (iv)) in Proposition~\ref{thm:DFE} is satisfied, in which case $x^*$ must be stable (by Proposition~\ref{thm:DFE}); otherwise, by statement~(v) in Proposition~\ref{thm:DFE} the point $x^*$ is unstable. This concludes the proof of statement~i) and implies the uniqueness of statement~ii).


\textit{Proof of statement~\ref{q2}}:
We show that if $x^*$ is locally stable it is almost globally asymptotically stable with domain of attraction $\mathcal{D}$.

Consider $z_i, i\in[2]$. From equation~\ref{eq:main}, for generic choices of $c_i, c_\mcd$, we have the analytic solution:
\begin{align*}
    z_i(t) = \frac{\exp((c_i - c_\mcd)t)}{z_i^{-1}(0) + \exp((c_i - c_\mcd)t)}
\end{align*}
Thus, $z_i$ is monotonic and approaches $z_i^* = [c_i > c_\mcd]$ exponentially fast, if $z_i(0)\in(0,1)$.
Additionally, if $z_i(0) = z_i^* \in\{0,1\}$ then $\dot{z}_i = 0$.
Therefore, $\epsilon_i := z_i - z_i^* \to 0$ as $t \to \infty$ for $z_i(0)\in(0,1)\cup\{z_i^*\}$.

Observe that, for any $x(t) \in \Gamma$, $\max(y_i) = 1$, $\max(1 - y_1 - y_2) = 1$, and $\max(q_{z_\S}) = 1$, due to Assumption~\ref{assume:q}. Therefore, for $y_i, i\in[2]$, we obtain the following upper bound:
\begin{align}\label{key:ineq_y}
\dot y_i
&= (\beta_i (1 - y_1 - y_2)q_{z_\S}(1 - z_i (1 - q)) - \delta_i) y_i \nonumber \\
&\le (\beta_i (q_{[c_i>c_\mcd]} + \epsilon_i(1 - q)) - \delta_i) y_i \nonumber \\
&= (q_{[c_i>c_\mcd]} \beta_i - \delta_i + \epsilon_i \beta_i (1 - q)) y_i
\end{align}
Under the hypothesis, $q_{[c_i>c_\mcd]}\beta_i / \delta_i < 1$, there exists $t_0\in[0,\infty)$ s.t. $\epsilon_i$ is small enough to guarantee that $\dot{y}_i < 0$ for all subsequent time (note: the behavior of $y_i$ before such time is irrelevant since $y \in \Delta, \forall t$). Therefore, $y \to \bs{0}$ as $t \to \infty$; moreover, since $q_{[c_i>c_\mcd]} \beta_i - \delta_i + \epsilon_i \beta_i (1 - q) < -k$ for some constant $k$ for all $t\ge t_0$, we have the bound $y_i(t) \le y_i(0) \exp(-kt)$, and $y_i$ decays exponentially fast.
Let $t_1 > t_0$ be large enough s.t. $y_i(t_1), i\in[2]$ are small enough so that $r^\top y < c_\mcd$. Then, $\forall t \ge t_1$, $\dot z_\S = z_\S (1 - z_\S) (2 \textbf{r}^\top\textbf{y} - c_\mathcal{D}) < 0$, provided that $z_\S(t_1)\in(0,1)$.
If $z_\S(0) = 0$ there is no issue, since it will remain at $0$ (we ignore $z_\S(0) = 1$ since it always precludes convergence and excludes us from $\mathcal{D}$).
Thus, for $z_\S$ to converge to $0$, it is sufficient show that $z_\S(0)\in(0,1)$ implies $z_\S(t) < 1, \forall t$.
Observe that $2r^\top y$ is less than the constant $2\max(r_1, r_2)$. We upper bound $\dot{z}_\S$ and solve the ODE:
\begin{align}\label{key:ineq_zs}
\dot{z}_\S(t) &= z_\S(t)(1 - z_\S(t))(2(r_1 y_1(t) + r_2 y_2(t)) - c_\mcd) \nonumber \\
\dot{z}_\S(t) &\le z_\S(t)(1 - z_\S(t))(2\max(r_1,r_2) - c_\mcd) \nonumber \\
z_\S(t) &\le \frac{\exp((2\max(r_1,r_2) - c_\mcd)t)}{\exp((2\max(r_1,r_2) - c_\mcd)t) + z_\S^{-1}(0) - 1}
\end{align}
Observe that $z_\S(0)\in(0,1) \implies z_\S^{-1}(0) - 1 > 0$. Then, inequality~\eqref{key:ineq_zs} gives $z_\S(0)\in(0,1) \implies z_\S(t) < 1, \forall t < \infty$; moreover, by similar argumentation as for $z_i$, $z_\S$ converges exponentially fast. Therefore, $z_\S$ will not hit the upper boundary in finite time and will be in the interior at $t_1$, and $z_\S \to 0$ as $t\to\infty$.\\
Thus, when $q_{[c_i > c_\mcd]} \beta_i / \delta_i < 1, i\in[2]$, $z_i(0) \ne [c_i < c_\mcd]$, and $z_\S(0) < 1$, we have $\lim_{t\to\infty} x = [0,0,0,[c_1>c_\mcd],[c_2>c_\mcd]] = x^*$, and $x^*$ is exponentially stable with domain of attraction $\mathcal{D}$.~\qed

We have the following remarks.

\begin{remark}\label{rem:almost:global}[Almost global]
The domain of attraction $\mathcal{D}$, from equation~\eqref{eq:D}, in Theorem~\ref{thm:globalstability:DFE0} is said to be \emph{almost global}, since the set $\Gamma - \mathcal{D}$ has measure zero. The domain, $\mathcal{D}$ encompasses the whole interior of $\Gamma$ as well as part of its boundary $\partial \Gamma$. Consequently, Theorem~\ref{thm:globalstability:DFE0} guarantees almost global exponential stability of the DFE in question.
\end{remark}

\begin{remark}\label{rem:p_DFE0}[Epidemiological interpretation] From an epidemiological viewpoint, Theorem~\ref{thm:globalstability:DFE0} says that as long as the reproduction numbers scaled by the interaction factor are less than unity (and exposing others to the virus is socially costly), we expect infected individuals to engage in social distancing, then both viruses to die out, and then for healthy people to stop social distancing.
\end{remark}

As the findings in Section~\ref{sec:arbitrary:costs} and Section~\ref{sec:stability:DFE} highlight, the admission of arbitrary costs in system~\eqref{eq:main} leads to a much more complicated analysis for the existence and stability of various 
\seb{equilibria} of system~\eqref{eq:main}, Consequently, from hereon, we will restrict the possible costs that system~\eqref{eq:main} admits. In particular, we have the following Assumption.
\begin{assumption}\label{assume:c}
    The costs satisfy the following: $c_i > c_\mcd$ for $i=1,2$.
\end{assumption}
This assumption asserts that the cost of socializing while infected is higher than the cost associated with not social distancing. This scenario is plausible when either local authorities implement strict lockdown policies, or when a community's social norms discourage social activity when sick.

\section{Existence of non-DFE equilibria}
\label{sec:analysis}
In this section, we identify all the non-DFE 
equilibria
of 
system \eqref{eq:main}. From  Assumption \ref{assume:c} ($c_i > c_\mcd$), we have $\dot{z}_i > 0$ at any state in $\interior\Gamma$. In other words, infected individuals never have an incentive to not practice social distancing. As a result, any 
\seb{equilibrium} with $z_i = 0$ for any $i=1,2$ cannot be locally asymptotically stable. Therefore, we will rule these out and restrict attention to 
\seb{equilibria} of the form $(y_1,y_2,z_S,1,1)$, with $y_1,y_2,z_S \in [0,1]$. 


In order to characterize the full set of 
\seb{equilibria}, we make use of the following functions.
\begin{align}
    \begin{aligned}
        h_1(y_1,y_2,z_\S,z_1) &:= - \delta_1+ \beta_1(1-y_1-y_2) \times 
        (qz_\S + (1-z_\S))(qz_1 + (1-z_1))   \\
        h_2(y_1,y_2,z_\S,z_2) &:= - \delta_2+ \beta_2(1-y_1-y_2) 
        (qz_\S + (1-z_\S))(qz_2 + (1-z_2))   \\
        h_\S(y_1,y_2) &:= 2(r_1y_1+r_2y_2)-c_\mcd
    \end{aligned}
\end{align}
%
%

In the next Lemma, we identify all unilateral equilibria in the system \eqref{eq:main}, as well as conditions on the parameters for which each one lies in the state space $\Gamma$.

\begin{lemma}[\textbf{Unilateral Equilibria}]\label{lem:unilateral}
    Consider system~\eqref{eq:main} under Assumptions~\ref{assume:betadelta}-\ref{assume:r}.
    With $i\in[2]$, the following characterizes all unilateral equilibria of the system:

    \noindent 1) The equilibrium $\bs{p}_{i0}$, defined by $y_i = 1 - \frac{\delta_i}{q\beta_i}$, $y_{3-i} = 0$, $z_S = 0$, and $z_1=z_2=1$, exists if and only if $\frac{\delta_i}{q\beta_i} < 1$.

    \noindent 2) The equilibrium $\bs{p}_{i1}$, defined by $y_i = 1 - \frac{\delta_i}{q^2\beta_i}$, $y_{3-i} = 0$, $z_S = 0$, and $z_1=z_2=1$, exists if and only if $\frac{\delta_i}{q^2\beta_i} < 1$.

    \noindent 3) The equilibrium $\bs{p}_{i\S}$, defined by $y_i = \frac{c_\mcd}{2r_i}$, $y_{3-i} = 0$, $z_S = \frac{1}{1-q} - \frac{\delta_i}{\beta_i(1-\frac{c_\mcd}{2r_i})q(1-q)}$, and $z_1=z_2=1$, exists if and only if $\frac{c_\mcd}{2r_i} \leq 1$ and
    \begin{equation}\label{eq:p1S_condition}
        q \left(1-\frac{c_\mcd}{2r_i}\right)^{-1} < q^2 R_i^0 < \left(1-\frac{c_\mcd}{2r_i}\right)^{-1}.
    \end{equation}
\end{lemma}
\textit{Proof:} We focus on proving the case $i=1$, as the case $i=2$ will follow completely analogous arguments. Thus, our task is to identify all 
\seb{equilibria} with $y_1 > 0$ and $y_2 = 0$. We prove each part separately.
    
    \noindent(1) In this part, we suppose that $z_S = 0$, which sets $\dot{z}_S = 0$. We then need to solve $h_1(y_1,0,0,1) = \beta_1(1-y_1) q - \delta_1 = 0$, resulting in $y_1 = 1 - \frac{\delta_1}{q\beta_1}$. It holds that $y_1 > 0$ if and only if $\frac{\delta_i}{q\beta_i} < 1$. 

    \noindent(2) In this part, we suppose that $z_S = 1$, which gives $\dot{z}_S = 0$. We then need to solve $h_1(y_1,0,1,1) = \beta_1(1-y_1) q^2 - \delta_1 = 0$, resulting in $y_1 = 1 - \frac{\delta_1}{q^2\beta_1}$. It holds that $y_1 > 0$ if and only if $\frac{\delta_i}{q^2\beta_i} < 1$.

    \noindent(3) In this part, we suppose that $z_S \in (0,1)$. In order for $\dot{z}_S = 0$, we need that $h_S(y_1,0) = 0$, which gives $y_1 = \frac{c_\mcd}{2r_1}$. For $y_1 \in (0,1]$, it is required that $\frac{c_\mcd}{2r_1} \leq 1$. In order for $\dot{y}_1 = 0$, we need to solve $h(\frac{c_\mcd}{2r_1},0,z_S,1) = 0$, which yields $z_S = \frac{1}{1-q} - \frac{\delta_1}{\beta_1(1-\frac{c_\mcd}{2r_1})q(1-q)}$. For $z_S > 0$, it is required that $\frac{\delta_i}{q\beta_i} < (1-\frac{c_\mcd}{2r_i})$. For $z_S < 1$, it is required that $q (1-\frac{c_\mcd}{2r_i}) < \frac{\delta_i}{q\beta_i}$.~$\square$

The equilibrium $\bs{p}_{i0}$ indicates a unilateral endemic state in which virus $i$ survives, and nobody in the population is social distancing. The equilibrium $\bs{p}_{i1}$ indicates a unilateral endemic state in which virus $i$ survives and everybody in the population is social distancing. The equilibrium $\bs{p}_{iS}$ indicates a unilateral endemic state in which virus $i$ survives and a fraction of the population is social distancing.

In the next Lemma, we identify the set of all coexistience equilibria and conditions for when they exist in the state space $\Gamma$. 

\begin{lemma}[\textbf{Coexistence equilibria}]\label{lem:coexist}
    Consider system~\eqref{eq:main} under Assumptions \ref{assume:betadelta}-\ref{assume:r}.
    Coexistence equilibria in system \eqref{eq:main} can exist only if $R^0_1=R^0_2=:R^0$. 

    \vspace{1mm}
    
    \noindent 1) The line, $\mcl_0$, of coexistence equilibria defined as follows:
    \begin{equation}\label{eq:line0}
        \mathcal{L}_0 := \left\{ (y_1,y_2,0,1,1) \;\middle|\; y_1 + y_2 = 1 - \frac{1}{q R^0} \right\},
    \end{equation}
    exists if and only if $qR^0 > 1$.

    \vspace{1mm}
    
    \noindent 2) The line, $\mcl_1$, of coexistence equilibria defined as follows:
    \begin{equation}\label{eq:line1}
        \mathcal{L}_1 := \left\{ (y_1,y_2,1,1,1) \;\middle|\; y_1 + y_2 = 1 - \frac{1}{q^2R^0} \right\}
    \end{equation}
    exists if and only if $q^2 R^0 > 1$.

    \vspace{1mm}
    
    \noindent 3) The curve, $\mcl_\S$, of coexistence equilibria defined as follows:
    Let $\mcal{L}_S$ be the set of points of the form $(y_1,y_2,z_S,1,1)$, parameterized by the value $y_1$, where $y_2 = \frac{c_\mcd}{2r_2} - \frac{r_1}{r_2}y_1$,  $z_S = \frac{1}{1-q} - \frac{1}{q(1-q)R^0(1-y_1-y_2)}$, and $y_1$ lies in the range
    \begin{equation}\label{eq:lineS}
        \underline{B} < y_1 < \bar{B},
    \end{equation}
    where
    \begin{align}
        \underline{B}:&=\max\left\{ 0, \frac{r_2}{r_1}\left( \frac{c_\mcd}{2r_2} - 1 \right), \ \frac{1 - \frac{c_\mcd}{2r_2} - \frac{1}{q^2R^0}}{1 - \frac{r_1}{r_2}} \right\} \label{eq:lineS:lower} \\
    \text{ and  } 
        \bar{B} :&= \min\left\{ \frac{c_\mcd}{2r_1}, \ \frac{1 - \frac{c_\mcd}{2r_2}}{1 - \frac{r_1}{r_2}}, \ \frac{1 - \frac{c_\mcd}{2r_2} - \frac{1}{qR^0}}{1 - \frac{r_1}{r_2}}, 1 \right\}. \label{eq:lineS:upper}
    \end{align}
    Then $\mcal{L}_S$ is a continuum of equilibria contained in the state space $\Gamma$ if and only if $\underline{B} < \bar{B}$.
\end{lemma}
\textit{Proof:}
    We prove each case separately.

    \noindent(1) Suppose that $z_S = 0$. For $\dot{y}_1 = 0$, we solve $h_1(y_1,y_2,0,1) = \beta_1(1-y_1-y_2) \cdot q - \delta_1 = 0$, resulting in $y_1 + y_2 = 1 - \frac{\delta_1}{q\beta_1}$. For $\dot{y}_2 = 0$, we solve $h_2(y_1,y_2,0,1) = 0$, resulting in $y_1 + y_2 = 1 - \frac{\delta_2}{q\beta_2}$. In order for an 
    \seb{equilibrium} with $y_1,y_2 > 0$ satisfying these equalities to exist, it is required that $\frac{\delta_1}{\beta_1} = \frac{\delta_2}{\beta_2}$. Then, the line of equilibria $\mcal{L}_0$ exists if and only if $0<1-\frac{1}{qR^0} < 1$. Since $R^0, q > 0$ by Assumptions \ref{assume:betadelta} and \ref{assume:q}, this is equivalent to $qR^0 > 1$.

    \noindent(2) This case is analogous to case (1), and so we omit these details for brevity.

    \noindent(3) Here, suppose that $z_\S \in (0,1)$. For $\dot{z}_S = 0$, it holds that $h_S(y_1,y_2) = 0$, or that $y_2 = \frac{c_\mcd}{2r_2} - \frac{r_1}{r_2}y_1$. In order for $y_2 \in (0,1)$, it must hold that
    \begin{equation}\label{eq:coex_range1}
        \frac{r_2}{r_1}\left( \frac{c_\mcd}{2r_2} - 1 \right) < y_1 <    \frac{c_\mcd}{2r_1}.
    \end{equation}
    
    Let us denote $\ell(y_1) := 1 - y_1 - y_2 = (1-\frac{c_\mcd}{2r_2}) - y_1(1-\frac{r_1}{r_2})$. Note it must hold that $\ell(y_1) \geq 0$, or 
    \begin{equation}\label{eq:coex_range2}
        y_1 \leq \frac{1 - \frac{c_\mcd}{2r_2}}{1 - \frac{r_1}{r_2}}.
    \end{equation}

    For $\dot{y}_1 = 0$, we solve $h_1(y_1,y_2,z_S,1) = 0$, which yields $z_\S = \frac{1}{1-q} - \frac{\delta_1}{q(1-q)\beta_1 \ell(y_1)}$. Similarly, for $\dot{y}_2 = 0$, we solve $h_2(y_1,y_2,z_S,1) = 0$, which yields $z_\S = \frac{1}{1-q} - \frac{\delta_2}{q(1-q)\beta_2 \ell(y_1)}$. For these two equations to be satisfied, it is required that $\frac{\delta_1}{\beta_1} = \frac{\delta_2}{\beta_2}$.

    Now, the condition that $z_S \in (0,1)$ is equivalent to
    \begin{equation}\label{eq:coex_range3}
        \frac{1 - \frac{c_\mcd}{2r_2} - \frac{1}{q^2R^0}}{1 - \frac{r_1}{r_2}} < y_1 < \frac{1 - \frac{c_\mcd}{2r_2} - \frac{1}{qR^0}}{1 - \frac{r_1}{r_2}}.
    \end{equation}

    Putting together all conditions on $y_1$ \eqref{eq:coex_range1}, \eqref{eq:coex_range2}, and \eqref{eq:coex_range3}, in addition to the restriction $y_1 \in (0,1)$, we conclude that a continuum of equilibria defined by $\mcal{L}_S$ lies in the state space $\Gamma$ if and only if $\underline{B} < \bar{B}$. ~$\square$\\   
%
Lemma \ref{lem:coexist} asserts that there are no isolated coexistence equilibrium. there are either no coexistence equilibria, or an infinite number of them situated on continuum of coexistence equilibria. We note that the conditions for the existence of the three lines are not mutually exclusive. We, also, remark that coexistence equilibria can only exist when the reproduction numbers of both viruses are identical, i.e. $R^0 = \beta_1/\delta_1 = \beta_2/\delta_2$.
Interestingly, this necessary condition does not depend on the risk perception parameters $r_1,r_2$.

\vspace{-2mm}

\section{Stability analysis of unilateral equilibria}\label{sec:unilateral:FPs}
In this section, we investigate the stability of the unilateral \seb{equilibrium}, $\bs{p}_{10} = [1 - \frac{\delta_1}{q\beta_1},0,0,1,1]^\top$. 

Let $R_i := q_{z_\S(t)} q_{z_i(t)} \frac{\beta_i}{\delta_i}$ be the effective reproduction number of virus $i$ at time $t$.
Let $R_i^* := q_{z_\S^*} q_{z_i^*} \frac{\beta_i}{\delta_i}$ be the effective reproduction number at some equilibrium point $x^*$.

We define the following domain:
\begin{equation}\label{eq:D_p_10}
\mathcal{D}_{\bs{p}_{10}} := \Delta\times(0,1]\times[0,1]\times[0,1)\times(0,1]^2
\end{equation}
Note that $\Gamma\supset\mathcal{D}_{\bs{p}_{10}}\supset\interior\Gamma$.  We have the following result

\begin{theorem}\label{thm:fp:global}
   Consider system~\eqref{eq:main} under Assumptions~\ref{assume:betadelta} and~\ref{assume:c}.
   The equilibrium point $\bs{p}_{10} = [1 - (R_1^*)^{-1},0,0,1,1]^\top$ exists, if $R_1^* > 1$. If $\max(1, R_2^*) < R_1^*$ and $R_1^* < (1 - \frac{c_\mcd}{2 r_1})^{-1}$, it is exponentially stable with domain of attraction $\mathcal{D}_{\bs{p}_{10}}$ (as defined in \ref{eq:D_p_10}); otherwise, it is unstable.
\end{theorem}
\textit{Proof:} By hypothesis, $1 - (R_1^*)^{-1} \in(0,1)$, which by Lemma~\ref{lem:unilateral}, implies the existence of the equilibrium point $\bs{p}_{10}$.

Under Assumption~\ref{assume:c}, $z_i, i\in[2]$ approach $z_i^* = 1$ exponentially.
Let $t_1$ be large enough s.t. $\epsilon_{z_1} := z_1 - z_1^*$ is suitably small, and let $\Delta y_1 := y_1 - y_1^*$. Then, from the first line in equation~\eqref{eq:main}, we have the following constraints on $\dot{y}_1$:
\begin{align*}
\dot{y}_1
&= (R_1 s - 1) \delta_1 y_1 \\
&= (R_1 (1 - y_1 - y_2) - 1) \delta_1 y_1 \\
&= (R_1 (1 - (1 - R_1^{*-1} + \Delta y_1) - y_2) - 1) \delta_1 y_1 \\
&= - R_1 \delta_1 ((1 - R_1 / R_1^*) + \Delta y_1 + y_2 + O(\epsilon_{z_1})) y_1 \\
&\le - R_1 \delta_1 (\Delta y_1 + y_2 + O(\epsilon_{z_1})) y_1 
\end{align*}
Since $O(\epsilon_{z_1})$ is taken to be negligible, while $y_1 + y_2 < y_1^*$, we have $\dot{y}_1 > 0$. Alternately, $y_1 + y_2 > y_1^*$ implies $\dot{y}_1 < 0$. Clearly, this yields the upper bound, $\lim_{t\to\infty} y_1 \le y_1^*$.

Let $t_2 > t_1$ be large enough such that $\forall t \ge t_2, y_1(t) \le y^*$ as shown above and recall the assumption $R_2^* < R_1^*$. Using the second line of equation~\eqref{eq:main}, we have the following:
\begin{align}
\dot y_2 &= (R_2 s - 1) \delta_2 y_2  \nonumber \\
&\le (R_2 (1 - (1 - R_1^{*-1}) - y_2) - 1) \delta_2 y_2, \label{ineq:global:p10:y2:1}\\
&= (R_2 R_1^{*-1} - 1 - R_2 y_2) \delta_2 y_2 \nonumber \\
&\le -R_2 \delta_2 y_2^2, \label{ineq:global:p10:y2:2}
\end{align}
We substitute $y_1 = y_1^* = 1 - R_1^{*-1}$ in the second line of equation~\eqref{eq:main}, which, since $y^*$ is the maximum value that $y_1$ can take,  means that $y_1$ is upper bounded by $y^*$, and, therefore, we have inequality~\eqref{ineq:global:p10:y2:1}.
Thereafter, since by assumption $R_2^* < R_1^*$, we have $R_2 / R_1^* \le q$, and, since, by Assumption~\ref{assume:q}, $q \in (0,1)$, it must be that 
$R_2/R_1^* - 1 < 0$; yielding inequality~\eqref{ineq:global:p10:y2:2}. 
Observe that inequality~\eqref{ineq:global:p10:y2:2} holds for all $t \ge t_2$, and, together with Assumption~\ref{assume:betadelta}, implies $\lim_{t\to\infty} y_2 = 0$.

Next, we obtain a lower bound for $\dot{y}_1$ and solve the differential equation (in the first line of equation~\eqref{eq:main}) to get a lower bound for $y_1(t)$:
\begin{align}
\dot y_1 &= (R_1 s - 1) \delta_1 y_1 \nonumber \\
&\ge -\delta_1 y_1 \label{ineq:global:p10:2}\\
y_1(t) &\ge y_1(0) \exp(-\delta_1 t), \nonumber 
\end{align}
where inequality~\eqref{ineq:global:p10:2} is due to the fact that $s \in [0,1]$. 
Therefore, if $y_1(0) > 0$ then $y_1(t) > 0$ for all $t\in[0,\infty)$ (this is necessary to say that $y_1$ does not go extinct in the meantime). Let $t_3 > t_2$ be large enough such that $y_2 \approx 0$. Then, if $y_1 < y_1^*$ we have $\dot{y}_1 > 0$ as shown above, and so $\lim_{t\to\infty} y_1 = y_1^*$.
Therefore, if $\max(1, R_2^*) < R_1^*$ and $y_1(0) > 0$, then both the viruses will converge to $y^* = [1 - R_1^{*-1}, 0]$.

Let $t_4 > t_3$ be large enough s.t. $\epsilon_{R_1} := R_1 - R_1^*$ and $|y - y^*|$ are suitably small. Then, if $R_1^* < (1 - c_\mcd / 2 r_1)^{-1}$, $z_\S$ will converge to $0$, since:
\begin{align*}
\dot{z_\S} &= z_\S (1 - z_\S) (2 \textbf{r}\textbf{y} - c_\mathcal{D}) \\
&= z_\S (1 - z_\S) (2 r_1(1 - R_1^{*-1}) - c_\mathcal{D}) \\
&< z_\S (1 - z_\S) (2 r_1(1 - (1 - \frac{c_\mathcal{D}}{2 r_1})) - c_\mathcal{D}) \\
&= z_\S (1 - z_\S) \cdot 0 \\
\dot{z_\S} &< 0 \label{ineq:global:p10:z_s}
\end{align*}
Therefore, $\lim_{t\to\infty} z_\S = 0$. Thus, if $y_1(0) > 0$, $z_\S(0) < 1$, and $z_i > 0, i\in[2]$, 
the system converges to $\bs{p}_{10}$.

We now show that if the hypothesis of the theorem is not met, $\bs{p}_{10}$ is unstable.
The Jacobian evaluated at $\bs{p}_{10}$ reads as in~\eqref{eq:jacobian:p10}.
\begin{equation}\label{eq:jacobian:p10}
J(\bs{p}_{10}) = \scriptsize
\begin{bmatrix}
    \delta_1 - \beta_1 q &
    \delta_1 - \beta_1 q &
    -\delta_1 (1 - \frac{\delta_1}{q\beta_1}) (1 - q) &
    -\delta_1 (1 - \frac{\delta_1}{q\beta_1}) \frac{(1 - q)}{q} &
    0 \\
    0 & \beta_2 \frac{\delta_1}{\beta_1} - \delta_2 & 0 & 0 & 0 \\
    0 & 0 & 2 r_1 (1 - \frac{\delta_1}{q\beta_1}) - c_\mathcal{D} & 0 & 0 \\
    0 & 0 & 0 & -(c_1 - c_\mathcal{D}) & 0 \\
    0 & 0 & 0 & 0 & -(c_2 - c_\mathcal{D}) \\
\end{bmatrix}
\end{equation}
Note that $ J(\bs{p}_{10})$ is upper triangular; its eigenvalues are the entries along its diagonal.
If $R_1^* < 1$, it is clear that $J(\bs{p}_{10})_{11} > 0$, and instability of the equilibrium follows from \cite[Theorem~4.7, statement ~ii)]{khalil2002nonlinear}.
Similarly, if $R_1^* < R_2^*$ (resp. $R_1^* > (1 - c_\mcd/(2r_1))^{-1}$), then $J(\bs{p}_{10})_{22} > 0$ (resp. $J(\bs{p}_{10})_{33} > 0$) and the equilibrium is unstable.
$\square$

\begin{remark}\label{rem:almost:global:p10}[Almost global]
The domain of attraction $\mathcal{D}_{\bs{p}_{10}}$, from equation (\ref{eq:D_p_10}), in Theorem~\ref{thm:fp:global} is almost global for reasons analogous to those in Remark~\ref{rem:almost:global}.
Consequently, Theorem~\ref{thm:fp:global} guarantees almost global exponential  stability of $\bs{p}_{10}$.
\end{remark}

We define the following domain:
\begin{equation}\label{eq:D_p_11}
\mathcal{D}_{\bs{p}_{11}} := (0,1]\times[0,1]\times[0,1)\times[0,1)^2
\end{equation}
Note that $\Gamma\supset\mathcal{D}_{\bs{p}_{11}}\supset\interior\Gamma$.

\begin{proposition}\label{prop:fp2:global}
   Consider system~\eqref{eq:main} under Assumptions~\ref{assume:betadelta} and~\ref{assume:c}.
   The equilibrium point $\bs{p}_{11} = [1 - (R_1^*)^{-1},0,1,1,1]^\top$ exists, if $R_1^* > 1$. If $\max(1, R_2^*) < R_1^*$ and $R_1^* > (1 - \frac{c_\mcd}{2 r_1})^{-1}$, then $\bs{p}_{11}$ is stable; otherwise, it is unstable.
   Furthermore, if $R_2^0 < q R_1^0$ then it is almost globally exponentially stable with domain of attraction $\mathcal{D}_{\bs{p}_{11}}$ (as defined in \ref{eq:D_p_11})
\end{proposition}
The proof of Proposition~\ref{prop:fp2:global} follows almost identically to Theorem~\ref{thm:fp:global}. Local (in)stability can be adduced from the Jacobian. However, inequality \eqref{ineq:global:p10:y2:2} in Theorem~\ref{thm:fp:global} no longer holds uniformly. Thus for almost global stability we require $R_2/R_1^* < 1$ which holds by hypothesis that $R_2^0 < q R_1^0$.

Note that analogous results establishing stability for the 
\seb{equilibria} $\bs{p}_{20} = [0, 1 - \frac{\delta_2}{q\beta_2},0,1,1]$ and $\bs{p}_{21} = [0, 1 - \frac{\delta_2}{q^2\beta_2},0,1,1]$ can be obtained by means of suitable change of variables in Theorem~\ref{thm:fp:global} and Proposition~\ref{prop:fp2:global}, respectively.

We have the following remark.
\begin{remark}
    The condition ($R_1^* < (1 - \frac{c_\mcd}{2 r_1})^{-1}$) in Theorem~\ref{thm:fp:global}, and 
    the condition ($R_1^* > (1 - \frac{c_\mcd}{2 r_1})^{-1}$) in Proposition~\ref{prop:fp2:global} are mutually exclusive; it is not possible for a given system to satisfy both these conditions at the same time. Therefore, the equilibria $\bs{p}_{10}$ and $\bs{p}_{11}$ \emph{cannot} be simultaneously locally exponentially stable. In fact, if $\bs{p}_{10}$ (resp. $\bs{p}_{11}$) is locally exponentially stable, then $\bs{p}_{11}$ (resp. $\bs{p}_{10}$) is unstable.
\end{remark}

Next, we identify sufficient conditions for the (in)stability of $\bs{p}_{1\S}$. We have the following result.
\begin{proposition}\label{prop:p1s}
     Consider system~\eqref{eq:main} under Assumptions~\ref{assume:betadelta} and \ref{assume:c}.
     If, for each $i\in[2]$, $c_\mcd < 2 r_i$ and $q(1-\frac{c_\mcd}{2r_i})^{-1} < q^2 R_i^0 < (1-\frac{c_\mcd}{2r_i})^{-1}$, then the 
     \seb{equilibrium} $\bs{p}_{1\S}$ exists in $\Gamma$, and it is stable (resp. unstable) if $R_1^0 > R_2^0$ (resp. $R_1^0 < R_2^0$).
\end{proposition}
\textit{Proof:}
The assumption $c_\mcd < 2 r_i, i\in[2]$ is sufficient to satisfy the condition in  statement (iii) of Lemma~\ref{lem:unilateral}; thus, $\bs{p}_{1\S}$ is guaranteed to exist. The Jacobian evaluated at $\bs{p}_{1\S}$, post a suitable simplification, is as given in~\eqref{eq:jacobian:p1S}.
\begin{equation}\label{eq:jacobian:p1S}
J(\bs{p}_{1\S}) = \begin{bmatrix}
-\frac{\delta_1}{1-\frac{c_\mcd}{2r_1}} \cdot\frac{c_\mcd}{2r_1} &
-\frac{\delta_1}{1-\frac{c_\mcd}{2r_1}} \cdot\frac{c_\mcd}{2r_1} &
-\beta_1 \frac{c_\mcd}{2r_1} (1 - \frac{c_\mcd}{2r_1}) (1 - q) q &
-\delta_1\frac{c_\mcd}{2r_1} \frac{1 - q}{q} & 0 \\
0 & \frac{\beta_2}{\beta_1}\delta_1 - \delta_2 & 0 & 0 & 0 \\
z_S(1 - z_S) 2 r_1 & z_S(1 - z_S) 2 r_2 & 0 & 0 & 0\\
0 & 0 & 0 & -(c_1 - c_\mcd) & 0 \\
0 & 0 & 0 & 0 & -(c_2 - c_\mcd) \\
\end{bmatrix}
\end{equation}
We partition $J(\bs{p}_{1\S}) = [J_1(\bs{p}_{1\S}), J_2(\bs{p}_{1\S}); \bs{0}_{2\times3}, J_3(\bs{p}_{1\S})]$. Since this partitioning is upper triangular, the spectrum of the Jacobian is given by the spectra of $J_1(\bs{p}_{1\S})$ and $J_3(\bs{p}_{1\S})$.
Since $J_3(\bs{p}_{1\S})$ is diagonal, we immediately have the eigenvalues $\lambda_4 = -(c_1 - c_\mcd)$ and $\lambda_5 = -(c_2 - c_\mcd)$, which are both negative under Assumption~\ref{assume:c}.

Consider $\spec(J_1(\bs{p}_{1\S}))$. With $J_{ij} := [J_1(\bs{p}_{1\S})]_{ij}$, the characteristic equation of $J_1(\bs{p}_{1\S})$ is:
$$\det(J_1(\bs{p}_{1\S})- \lambda\I)
= (\lambda - J_{22})(\lambda(\lambda - J_{11}) - J_{13} J_{31}) = 0$$
Immediately, we have the eigenvalue $\lambda_2 = J_{22} = \frac{\beta_2}{\beta_1}\delta_1 - \delta_2$. Thus, by Assumption~\ref{assume:betadelta}, $\lambda_2 < 0 \iff \frac{\beta_1}{\delta_1} > \frac{\beta_2}{\delta_2}$ and $\lambda_2 > 0 \iff \frac{\beta_1}{\delta_1} < \frac{\beta_2}{\delta_2}$.

\noindent Applying the quadratic formula to the remaining factor yields:
$$\{\lambda_1, \lambda_3\} = \frac{1}{2} J_{11} \pm \sqrt{J_{11}^2 + 4 J_{13} J_{31}}$$
Observe that under our Assumptions~\ref{assume:betadelta} and \ref{assume:c} with $c_\mcd < 2r_1$, we have $J_{11} < 0$, $J_{13} < 0$. From our earlier application of Lemma~\ref{lem:unilateral} (iii), we know $z_\S\in(0,1)$; thus, since $c_\mcd < 2r_1$ and so $r_1 > 0$, we have $J_{31} > 0$ (see Equation~\ref{eq:jacobian:p1S}).
Therefore, $4J_{13}J_{31} < 0$. Thus, if $|J_{11}^2| > |4J_{13} J_{31}|$ then $\left|\sqrt{J_{11}^2 + 4J_{13}J_{31}}\right| < |J_{11}|$, and the larger eigenvalue is bounded by: $\lambda < \frac{1}{2} \left[J_{11} + |J_{11}|\right] = 0$, i.e., both eigenvalues are negative. Alternatively, if $|J_{11}^2| < |4J_{13} J_{31}|$, then the quantity under the radical is negative and real. Therefore, its roots are purely imaginary, and, since $J_{11} < 0$, the result has strictly negative real part. Therefore, $\lambda_1$ and $\lambda_3$ are always negative. Consequently, $\bs{p}_{1\S}$ is stable if $\beta_1/\delta_1 > \beta_2/\delta_2$; it is 
unstable if $\beta_1/\delta_1 < \beta_2/\delta_2$.~$\square$

\section{Stability of continua of coexistence equilibria} \label{sec:line}
In this section, we identify conditions 
for (in)stability of continua of coexistence equilibria, $\mathcal L_0$, $\mathcal L_1$, and $\mathcal L_s$. 

\subsection{Stability of line $\mathcal L_0$}
Our first main result is the following theorem.
\begin{theorem}\label{thm:l0}
    Consider system~\eqref{eq:main} under Assumptions~\ref{assume:betadelta}, \ref{assume:q}, and 
    \ref{assume:c}. 
   Suppose further that $R^0 := R_1^0 = R_2^0$ and $qR^0 > 1$.
    Then the line of equilibria $\mcl_0$ is locally exponentially stable if
    $qR^0 < \min_{i\in[2]}(1 -  \frac{c_\mathcal{D}}{2r_i})^{-1}$; it is unstable if $qR^0 > \min_{i\in[2]}(1 -  \frac{c_\mathcal{D}}{2r_i})^{-1}$.
\end{theorem}
\textit{Proof:} Under the hypothesis of the theorem, the conditions in Lemma~\ref{lem:coexist}, statement (1) are satisfied,  and, consequently,  $\mcl_0$ is guaranteed to exist. Equation~\eqref{eq:line0} states that $1 - y_1 - y_2 = s = \frac{1}{qR^0}$, and so $\beta_isq = \delta_i, i\in[2]$. We use this identity, as well as substituting, $z_S = 0, z_1 = z_2 = 1$, to simplify the Jacobian, as given in~\eqref{eq:jacobian:line:l0}. %
\begin{equation}\label{eq:jacobian:line:l0}     
J(\mcl_0) = \begin{bmatrix}
    -\beta_1 y_1 q & -\beta_1 y_1 q &
    -\beta_1 y_1 s (1 - q) q & -\beta_1 y_1 s (1 - q) & 0 \\
    -\beta_2 y_2 q & -\beta_2 y_2 q &
    -\beta_2 y_2 s (1 - q) q & 0 & -\beta_2 y_2 s (1 - q) \\
    0 & 0 & 2(r_1 y_1 + r_2 y_2) - c_\mathcal{D} & 0 & 0\\
    0 & 0 & 0 & -(c_1 -c_\mathcal{D}) & 0 \\
    0 & 0 & 0 & 0 & -(c_2 - c_\mathcal{D}) \\
\end{bmatrix}
\end{equation}
\normalsize
We partition $J(\mcl_0) = [J_1(\mcl_0), J_2(\mcl_0); \bs{0}_{3\times2}, J_3(\mcl_0)]$.
Observe that $J(\mcl_0)$ is a block upper triangular matrix, so its spectrum is given by the spectrum of $J_1(\mcl_0)$ and the spectrum of $J_3(\mcl_0)$.
Let $\lambda_i$ denote the $i^{th}$ eigenvalue of $J(\mcl_0)$.
Consider $J_1(\mcl_0)$, which has the structure $[a, a; b, b]$, where  $a = -\beta_1 y_1 q$ and $b = -\beta_1 y_2 q$. This yields the eigen pairs $\{(0, [1,-1]^\top), (a+b,[a/b,1]^\top\}$.
By Assumption~\ref{assume:betadelta}, $\beta_i,\delta_i > 0, i\in[2]$; with the assumption $qR^0 > 1$ we have $0 < \frac{1}{q R^0} < 1$. Then the line equation requires $y_1 + y_2 = 1 - \frac{1}{q R^0}\in(0,1)$. Therefore, since, by Assumption~\ref{assume:q}, $q \in (0,1]$, we have that $\lambda_2 = a + b = (-q(\beta_1 y_1 + \beta_2 y_2)) < 0$.
Next, we consider a perturbation along the direction of the eigenvector corresponding to $\lambda_1 = 0$: Let $\Delta y := \epsilon\cdot[1, -1]^\top$ for some arbitrarily small $|\epsilon|$. Observe that if $y^*\in\mcl_0$, then $y^* + \Delta y\in\mathcal{L}_0$. This follows from substitution into the line equation:
$1 - y_1 - y_2 = \frac{1}{qR^0}
\iff 1 - (y_1 + \epsilon) - (y_2 - \epsilon) = \frac{1}{qR^0}$.
Therefore, the null-space of $J_1(\mcl_0)$ corresponds to the line of equilibria.

Consider $J_3(\mcl_0) = \diag(2(r_1 y_1 + r_2 y_2) - c_\mathcal{D}, c_1 - c_\mcd, c_2 - c_\mcd)$. Being a diagonal matrix, the eigenvalues of $J_3(\mcl_0)$ are its diagonal entries. 
Thus, $\lambda_3 = 2(r_1 y_1 + r_2 y_2) - c_\mathcal{D}$, $\lambda_4 = c_\mcd - c_1$ and $\lambda_5 = c_\mcd - c_2$. Since by Assumption~\ref{assume:c}, $c_\mcd < c_i, i\in[2]$, we have $\lambda_4 < 0$ and $\lambda_5 < 0$.
Consider $\lambda_3$. For it to be negative, we require that $2(r_1 y_1 + r_2 y_2) - c_\mathcal{D} < 0$.
The line equation for $\mcl_0$ \eqref{eq:line0} gives $y_2 = 1 - \frac{1}{qR^0} - y_1$. Then, in terms of $y_1$, we get
$y_1 (r_1 - r_2) < \frac{c_\mathcal{D}}{2} - r_2 \left(1 - \frac{1}{qR^0}\right)$.
The line equation also gives the minimum and maximum values for $y_1$ as $0$ and $1 - \frac{1}{qR^0}$, respectively.
Substituting into the inequality and solving yields the condition: 
\begin{equation}\label{eq:cond:l0}
\left(1 - \frac{1}{qR^0} < \frac{c_\mathcal{D}}{2r_2}\right) \land
\left(1 - \frac{1}{qR^0} < \frac{c_\mathcal{D}}{2r_1}\right)
\end{equation}
\normalsize
By the theorem hypothesis, the condition in \eqref{eq:cond:l0} is satisfied. Therefore, $\lambda_3 < 0$. 
Thus, $\mcl_0$ is the center eigenspace of $J(\mcl_0)$, and the corresponding center manifold is stable \cite[Theorem~7.26]{sastry1999nonlinear}. 
If  $1 - \frac{1}{qR^0} > \frac{c_\mathcal{D}}{2r_1}$ or if $1-\frac{q}{R^0}>\frac{c_\mcd}{2r_2}$, then by a straightforward reversal of the arguments presented above, it can be seen that $\lambda_3>0$, which implies that $s(J(\mcl_0))>0$, which, from \cite[Theorem~5.42]{sastry1999nonlinear}, guarantees that the line $\mcl_0$ is unstable.~$\square$
\subsection{Stability of line $\mathcal L_1$}
Our second main result is the following theorem.
\begin{theorem}\label{thm:l1}
    Consider system~\eqref{eq:main} under Assumptions~\ref{assume:betadelta}, \ref{assume:q}, and \ref{assume:c}.
    Suppose further that $R^0 := R_1^0 = R_2^0$ and $q^2 R^0 > 1$.
    The line $\mcl_1$ is locally exponentially stable if
    $q^2 R^0 > \max_{i\in[2]}(1 -  \frac{c_\mathcal{D}}{2r_i})^{-1}$; it is unstable if $q^2 R^0 < \max_{i\in[2]}(1 -  \frac{c_\mathcal{D}}{2r_i})^{-1}$.
\end{theorem}
\textit{Proof:}
Under the hypothesis of the theorem, the conditions for Lemma~\ref{lem:coexist} statement (2) are satisfied; consequently, $\mcl_1$ is guaranteed to exist.
From Equation~\eqref{eq:line1}, we get $1 - y_1 - y_2 = s = (q^2 R^0)^{-1}$; note the additional $q$ in the denominator of the right-hand side, when compared to the right-hand side of~\eqref{eq:line0}. Under Assumptions~\ref{assume:betadelta},\ref{assume:q}, $(q^2 R^0)^{-1}\in(0,1)$.
The Jacobian for $\mcl_1$ can be put in terms of the Jacobian for $\mcl_0$: $J(\mcl_1) = \diag(q, q, -1, 1, 1) J(\mcl_0)$. Then, using an identical partitioning as in Theorem~\ref{thm:l0}, we get $J_1(\mcl_1) = q \cdot J_1(\mcl_0)$. Therefore, the eigenvalues of $J_1(\mcl_1)$ are the same as for $J_1(\mcl_0)$ scaled by $q$, and, since by Assumption~\ref{assume:q} $q$ is positive, do not differ in sign.
For $J_3(\mcl_1)$, we get the same eigenvalues as for $J_3(\mcl_0)$, except that the expression for $\lambda_3$ is negated.
This gives the reverse inequality to that in Theorem~\ref{thm:l0}; i.e.,  $2(r_1y_1 + r_2y_2) > c_\mcd$. By a similar derivation as in Theorem~\ref{thm:l0}, it is straightforward to show that if 
\begin{equation}\label{eq:cond:l1:stable}
\left(1 - \frac{1}{q^2R^0} > \frac{c_\mathcal{D}}{2r_2}\right) \land
\left(1 - \frac{1}{q^2R^0} > \frac{c_\mathcal{D}}{2r_1}\right)
\end{equation}
\normalsize
then $\lambda_3 < 0$.
The rest of the  proof follows identically to the proof of stability of line $\mcl_0$ in Theorem~\ref{thm:l0}. 
Analogously, if 
\begin{equation}\label{eq:cond:l1:unstable}
\left(1 - \frac{1}{q^2R^0} < \frac{c_\mathcal{D}}{2r_2}\right) \lor
\left(1 - \frac{1}{q^2R^0} < \frac{c_\mathcal{D}}{2r_1}\right),
\end{equation}
\normalsize
then it is clear that $\lambda_3 > 0$, which is sufficient for $\mcl_1$ to be unstable.~$\square$
\begin{remark}[Exclusivity of stability for $\mcl_0$ and $\mcl_1$]
Combining the stability conditions for $\mcl_0$ and $\mcl_1$ yields the inequality:
$q R^0 < (1 - \frac{c_\mathcal{D}}{2r_i})^{-1} < q^2 R^0, i\in[2]$,
which yields the condition: $q R^0 < q^2R _0$. Under Assumptions~\ref{assume:betadelta} and \ref{assume:q} this is unsatisfiable. Therefore, $\mcl_0$ and $\mcl_1$ cannot simultaneously be stable.~$\square$
\end{remark}
\subsection{Stability of the curve $\mcl_\S$} In this subsection, we identify a condition for local exponential stability of the continuum of coexistence equilibria, $\mcl_\S$.
\begin{theorem}\label{thm:ls}
    Consider system~\eqref{eq:main} under Assumptions~\ref{assume:betadelta}, \ref{assume:q}, and \ref{assume:c}. 
    Suppose further that $R^0 := R_1^0 = R_2^0$ and $qR^0 > 1$.
    If equations~\eqref{eq:lineS}-\eqref{eq:lineS:upper} in Lemma~\ref{lem:coexist} 
    hold, then $\mcl_\S$ is locally exponentially stable.
\end{theorem}
\textit{Proof:}
By Lemma~\ref{lem:coexist}, if equations~\eqref{eq:lineS}-\eqref{eq:lineS:upper} hold, then $\mcl_\S$ exists in $\Gamma$.
Recall that we 
parameterize $\mcl_S$ as follows:
$$y_2 = \frac{c_\mcd}{2r_2} - \frac{r_1}{r_2}y_1;
\text{\quad and\quad}
z_S = \frac{1}{1-q} - \frac{1}{q(1-q)R^0(1-y_1-y_2)}$$
Then, with our parameterization of $z_\S$, we 
 simplify $q_{z_\S}$:
\begin{align*}
q_{z_\S} = 1 - z_\S(1 - q)
= (q R^0 s)^{-1}
\end{align*}
With Assumption~\ref{assume:c}, $z_1^* = 1$ and so $q_{z_1^*} = q$; we can, hence, simplify the Jacobian at $\mcl_S$ yielding,
\begin{equation}\label{eq:jacobian:line:ls}  
J(\mcl_\S) = \begin{bmatrix}
-\delta_1 y_1 s^{-1} &
-\delta_1 y_1 s^{-1} &
-\beta_1 y_1 s (1 - q) q &
-\delta_1 y_1 \frac{1 - q}{q} &
0 \\
-\delta_2 y_2 s^{-1} &
-\delta_2 y_2 s^{-1} &
-\beta_2 y_2 s (1 - q) q &
0 &
-\delta_2 y_2 \frac{1 - q}{q} \\
z_S(1 - z_S) 2 r_1 & z_S(1 - z_S) 2 r_2 & 0 & 0 & 0 \\
0 & 0 & 0 & c_\mathcal{D} - c_1 & 0 \\
0 & 0 & 0 & 0 & c_\mathcal{D} - c_2 \\
\end{bmatrix}
\end{equation}
Note that the Jacobian can be written in the following block triangular form: $J(\mcl_S) = [\begin{smallmatrix}
A & B \\ \textbf{0}_{2\times3} & C
\end{smallmatrix}]$, where,
\begin{equation}\label{eq:jacobian:line:ls:A}  
A := \begin{bmatrix}
-\delta_1 y_1 s^{-1} &
-\delta_1 y_1 s^{-1} &
-\beta_1 y_1 s (1 - q) q \\
-\delta_2 y_2 s^{-1} &
-\delta_2 y_2 s^{-1} &
-\beta_2 y_2 s (1 - q) q \\
z_S(1 - z_S) 2 r_1 & z_S(1 - z_S) 2 r_2 & 0
\end{bmatrix};
C := \begin{bmatrix}
c_\mcd - c_1 & 0 \\
0 & c_\mcd - c_1
\end{bmatrix}
\end{equation}
and $B$ is conformal.
Thus, it is sufficient to consider the spectrum of $A$ and $C$.
Let $\lambda_1,\lambda_2,\lambda_3$ be the eigenvalues of $A$, and $\lambda_4, \lambda_5$ be the eigenvalues of $C$.
Immediately, $\lambda_4 = c_\mcd - c_1$ and $\lambda_5 = c_\mcd - c_2$, both of which are negative by Assumption~\ref{assume:c}.
Consider the characteristic polynomial of $A$:
\begin{align*}
\det(A - \lambda\I)
&= -\lambda^3
+ \lambda^2 (A_{11} + A_{21})
+ \lambda(A_{13} A_{31} + A_{23} A_{32}) \\
&+ A_{11} A_{23} (A_{31} - A_{32}) + A_{12} A_{21}(A_{32} - 1)
\end{align*}
By construction, we know that the center manifold of $\mcl_\S$ is $1$-dimensional \cite[Theorem 8.1]{khalil2002nonlinear}. Therefore, there is one eigenvalue with real part zero in the spectrum of $J(\mcl_\S)$. Furthermore, since there is exactly one---and complex eigenvalues of real matrices come in conjugate pairs---it is exactly $0$. Since, $\lambda_4,\lambda_5 \ne 0$, we know the zero eigenvalue corresponds to the spectrum of $A$.
Let $\lambda_3 = 0$ be this zero eigenvalue. Since $\lambda_3$ is a solution of the characteristic polynomial of $A$, we know the following:
\begin{align*}
\det(A - \lambda\I) &= 0 = A_{11} A_{23} (A_{31} - A_{32}) + A_{12} A_{21}(A_{32} - 1)
\end{align*}
That is, the zeroth order term in the characteristic polynomial of $A$ simplifies to zero. Therefore, we have, 
\begin{align*}
\det(A - \lambda\I)
&= -\lambda^3
+ \lambda^2 (A_{11} + A_{21})
+ \lambda(A_{12} A_{31} + A_{23} A_{32}) \\
&= -\lambda [ \lambda^2
- \lambda (A_{11} + A_{21})
- (A_{12} A_{31} + A_{23} A_{32})]
\end{align*}
Thus, the last two eigenvalues are given by the quadratic:
\begin{align*}
0 &= \lambda^2 - \lambda (A_{11} + A_{21}) - (A_{12} A_{31} + A_{23} A_{32}) \\
\{\lambda_1, \lambda_2\}
&= \frac{1}{2} \left[(A_{11} + A_{21})\pm\sqrt{(A_{11} + A_{21})^2 + 4(A_{12} A_{31} + A_{23} A_{32})}\right]
\end{align*}
Observe that on $\mcl_\S$, we have $s, y_1, y_2 > 0$. Recall also Assumptions \ref{assume:betadelta} and \ref{assume:q}. Thus, $A_{11}, A_{21} < 0$ and $A_{11} + A_{21} < 0$. Additionally, $A_{13}, A_{23} < 0$ and $A_{31}, A_{32} > 0$. Thus, $4(A_{13} A_{31} + A_{23} A_{32}) < 0$. Following similar argumentation as in the proof of Proposition~\ref{prop:fp2:global}, the remaining two eigenvalues must have negative real part.
Therefore, $\mcl_\S$ has a $4$-dimensional stable manifold and a $0$-dimensional unstable manifold and is locally exponentially stable \cite[Theorem~7.26]{sastry1999nonlinear}.

\section{Simulations}\label{sec:sims}
In this section, we first illustrate our theoretical findings via numerical simulations, and thereafter explore the effect of incorporating a co-ordination term into the dynamics which can be thought of as representing the possibility of social distancing being socially reinforced by those participating.
\subsection{Illustration of the theoretical findings}
For each type of equilibrium of system~\eqref{eq:main} identified in Section~\ref{sec:analysis},
we provide a brief description,
example parameters for its existence, and initial conditions
for which the dynamics of our system, where possible, converge to an equilibrium of interest; we provide several plots showing the system converging in simulation.

We select $x(0)$ in $\interior\Gamma$. Because the state variables often remain fixed at the boundary of their domain (e.g., $z_\S(t_k) = 1 \implies z_\S(t) = 1, \forall t \ge t_k$), we wish to demonstrate the system converging to the boundary equilibria, even when the system does not start on the boundary. 
Note that for stable equilibria we may always select initial conditions in $\interior\Gamma$ that converge to the 
said equilibria. Simulations were performed using fourth-order Runge-Kutta approximation with step size $h = 10^{-4}$.  
For ease of exposition, we write the indexed parameters in vector notation: $\beta := [\beta_1, \beta_2]^\top$, $\delta := [\delta_1, \delta_2]^\top$, $r := [r_1, r_2]^\top$, and $c := [c_1, c_2]^\top$.

For the first simulation, we aim to show convergence to $\bs{p}_{\DFE 0} := [0,0,0,1,1]^\top$.
To this end, we choose $\beta=[\begin{smallmatrix}
    0.8&&0.7
\end{smallmatrix}]$, $\delta=[\begin{smallmatrix}
    0.4&&0.4
\end{smallmatrix}]$, $r=[\begin{smallmatrix}
    3&&3
\end{smallmatrix}]$, $c=[\begin{smallmatrix}
    9&&9
\end{smallmatrix}]$, $c_{\mcd}=1$ and $q=0.5$.
With this choice of parameters, the conditions for Proposition~\ref{thm:DFE}, statement~iv) 
are satisfied. 
Let $x(0) = [0.4, 0.6, 0.1, 0.1, 0.2]^\top\in\interior\Gamma$. 
The 
\seb{equilibrium} $\bs{p}_{\DFE 0}$ is stable and attractive, as demonstrated in Figure~\ref{figure:P_DFE0}.
This shows the system converging to a DFE even though the basic reproduction numbers are both greater that unity, $\beta_i/\delta_i > 1, i\in[2]$ and the initial population is completely infected: $y_1(0) + y_2(0) = 1$.

\begin{figure}[h!]\centering
\begin{minipage}{0.45\textwidth}\centering
\includegraphics[width=.9\textwidth]{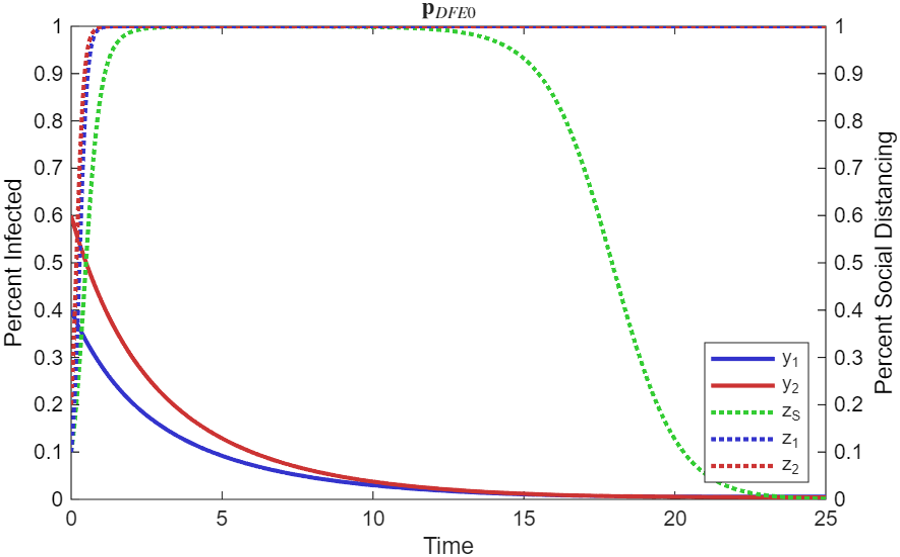}
\caption{
Simulation with $\bs{p}_{\DFE0}$ stable.
}
\label{figure:P_DFE0}
\end{minipage}
\begin{minipage}{0.45\textwidth}\centering
\includegraphics[width=.9\textwidth]{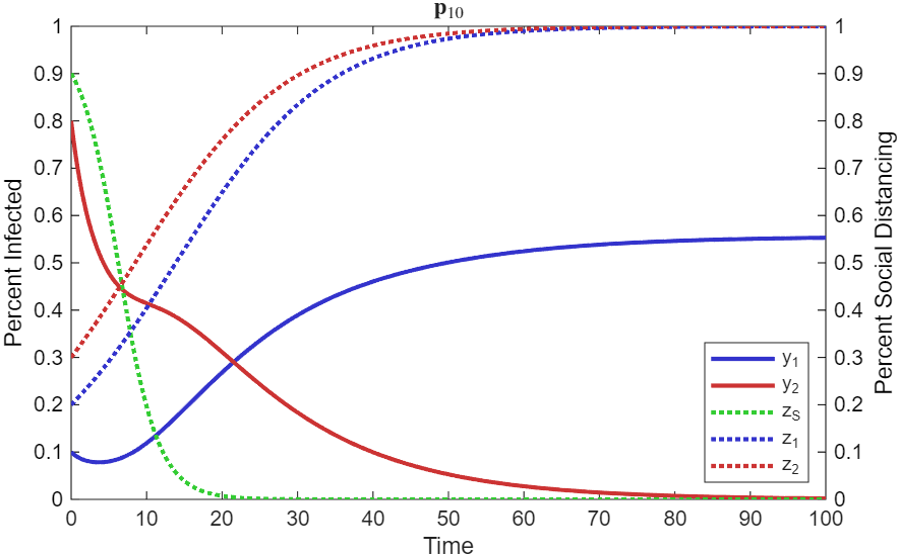}
\caption{
Simulation with $\bs{p}_{10}$ stable.
}
\label{figure:P_10}
\end{minipage}
\end{figure}




\noindent\textbf{Unilateral equilibrium with no SD:} For $\bs{p}_{10}$, we choose $\beta=[\begin{smallmatrix}
    0.9&&0.6
 \end{smallmatrix}]$, $\delta=[\begin{smallmatrix}
    0.2&&0.2
 \end{smallmatrix}]$, $r=[\begin{smallmatrix}
    0.2&&0.1
 \end{smallmatrix}]$, $c=[\begin{smallmatrix}
    0.6&&0.6
\end{smallmatrix}]$, $c_{\mcd}=0.5$ and $q=0.5$.
With this choice of parameters, $R_1^* = 9/4, R_2^* = 3/2$. Thus, $\max(1, R_2^*) < R_1^*$ is satisfied and
Theorem~\ref{thm:fp:global} implies convergence to $\bs{p}_{10} = [5/9,0,0,1,1]$.
Let $x(0) = [0.1, 0.8, 0.9, 0.2, 0.3]^\top\in\interior\Gamma$.
Consistent with Theorem~\ref{thm:fp:global}, we see that virus~$1$ is endemic while virus~$2$ becomes extinct, with no social distancing in the susceptible population; see Figure~\ref{figure:P_10}.

\noindent\textbf{Unilateral equilibrium with all SD:} For $\bs{p}_{11}$, we choose $\beta=[\begin{smallmatrix}
    2&&0.3
\end{smallmatrix}]$, $\delta=[\begin{smallmatrix}
    0.1&&0.6
\end{smallmatrix}]$, $r=[\begin{smallmatrix}
    0.6&&0.4
\end{smallmatrix}]$, $c=[\begin{smallmatrix}
    0.5&&0.5
\end{smallmatrix}]$, $c_{\mcd}=0.4$ and $q=0.3$.
With this choice of parameters, $R_1^* = 9/5, R_2^* = 9/200$. Thus, $\max(1, R_2^*) < R_1^*$ and $R_1^* > (1 - c_\mcd/2r_1) = 3/2$ are satisfied and Proposition~\ref{prop:fp2:global} implies convergence to $\bs{p}_{11} = [4/9,0,0,1,1]$.
Let $x(0) = [0.1, 0.9, 0.5, 0.4, 0.7]^\top\in\interior\Gamma$.
Indeed, in line with the result in Proposition~\ref{prop:fp2:global} we see that, with the entire healthy population social distancing, virus~$1$ is endemic while virus~$2$ has been eradicated; see Figure~\ref{figure:P_11}.


\begin{figure}[h!]\centering
\begin{minipage}{0.45\textwidth}\centering
\includegraphics[width=0.9\textwidth]{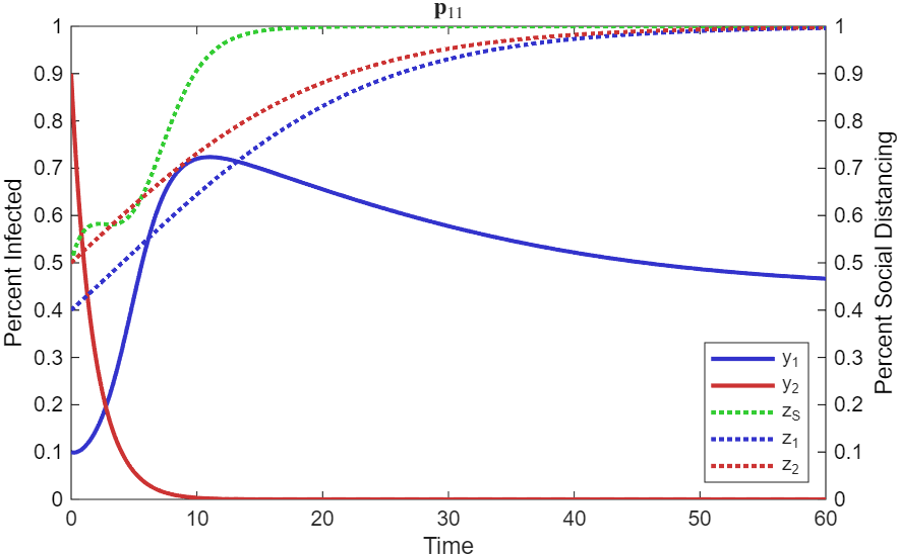}
\caption{
Simulation with $\bs{p}_{11}$ stable.
}
\label{figure:P_11}
\end{minipage}
\begin{minipage}{0.45\textwidth}\centering
\includegraphics[width=0.9\textwidth]{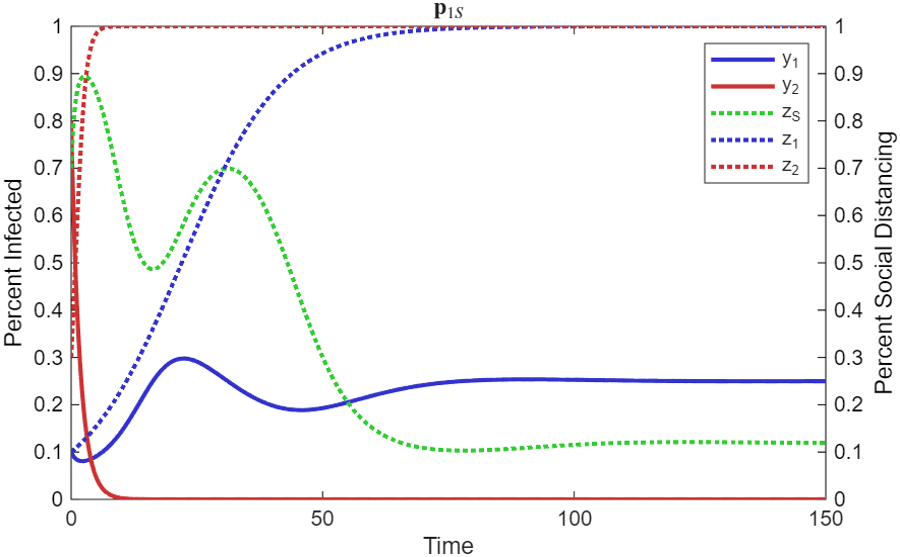}
\caption{
Simulation with $\bs{p}_{1\S}$ stable.
}
\label{figure:P_1S}
\end{minipage}
\end{figure}

\noindent\textbf{Unilateral equilibrium with some SD:}
From Lemma~\ref{lem:unilateral}, the existence of $\bs{p}_{1\S}$ requires $c_\mcd < 2 r_1$ so that $y_1\in(0,1)$ and $q(1-c_\mcd/(2r_i)) < (q\beta_i/\delta_i)^{-1} < (1-c_\mcd/(2r_i)), i\in[2]$ so that $z_\S\in(0,1)$.
We choose 
$\beta=[\begin{smallmatrix}
    0.7&&0.1
\end{smallmatrix}]$, $\delta=[\begin{smallmatrix}
    0.3&&0.6
\end{smallmatrix}]$, $r=[\begin{smallmatrix}
    1&&1
\end{smallmatrix}]$, $c=[\begin{smallmatrix}
    .6&&1.5
\end{smallmatrix}]$, $c_{\mcd}=0.5$ and $q=0.6$.
Observe that with such a choice of parameters, $q\beta_1/\delta_1 = 7/5$, $q\beta_2/\delta_2 = 1/10$, and $1 - c_\mcd/(2r_1) = 1 - c_\mcd/(2r_2) = 3/4$.
Then, the condition $q(1-c_\mcd/(2r_i)) < (q\beta_i/\delta_i)^{-1} < (1-c_\mcd/(2r_i))$ becomes $9/16 < 5/7 < 3/4$ which is true. Moreover, $\beta_1/\delta_1 > \beta_2/\delta_2$ so Proposition~\ref{prop:p1s} implies convergence to $\bs{p}_{1\S} = [1/4, 0, 5/42, 1, 1]^\top$.
Then, we can observe, in Figure~\ref{figure:P_1S}, virus~$1$ remains endemic while virus~$2$ dies out, with partial social distancing in the healthy population.

\noindent\textbf{Coexistence equilibria with no SD:}
For $\mcl_0$ to exist, Lemma~\ref{lem:coexist} requires that $R^0 = \beta_1/\delta_1 = \beta_1/\delta_1$ and $qR^0 > 1$.
We choose
$\beta=[\begin{smallmatrix}
    0.4&&0.4
\end{smallmatrix}]$, $\delta=[\begin{smallmatrix}
    0.2&&0.2
\end{smallmatrix}]$, $r=[\begin{smallmatrix}
    1&&2
\end{smallmatrix}]$, $c=[\begin{smallmatrix}
    4&&4
\end{smallmatrix}]$, $c_{\mcd}=2$ and $q=0.8$.
Such a choice guarantees the existence of a continuum of coexistence equilibria, specifically the line $\mcl_0$ (see Equation~\eqref{eq:line0}). Moreover, said choice also satisfies the condition in Theorem~\ref{thm:l0} for $\mcl_0$ to be stable.
Here we show multiple trajectories converging to $\mcl_0$.
We use the initial conditions, $x(0) \in\{ [0.1, 0.8, 0.9, 0.9, 0.9]^\top, [0.1, 0.1, 0.5, 0.5, 0.5]^\top, [0.8, 0.1, 0.5, 0.5, 0.5]^\top$, \break $[0.1, 0.1, 0.5, 0.5, 0.5]^\top\}$.
For clarity, we project onto $\Delta$; see Figure~\ref{figure:L_0:traj}.
The initial conditions are shown as red dots; the equilibria points are shown as cyan crosses, and the colored lines represent the trajectory of the system for each initial condition. The line $\mcl_0$ is shown as the dashed green line.
Note that the trajectory starting at $y = [.1,.8]^\top$ overshoots the line before converging.


\noindent\textbf{Coexistence equilibria with all SD:}
The $\mcl_{1}$ case is nearly identical to the $\mcl_{0}$ case. The $qR^0$ term gains a $q$ becoming $q^2 R^0 > 1$, and the stability condition is inverted to $2r^\top y > c_\mcd$.
Thus, we use similar parameters as for the $\mcl_0$, except we select $c_\mcd = 0.1$ so that $2 r_2 \left(1 - \frac{1}{qR^0}\right) > c_\mathcal{D}$.
Using the same initial conditions as for the $\mcl_0$ example, we observe similar behavior in Figure~\ref{figure:L_1:traj}. Observe that the line $\mcl_1$ has a lower $y_2$-intercept than the example for $\mcl_0$, even though they use the same $\beta$, $\delta$, and $q$ parameters; this is because $q^2 R^0 < q R^0$.
In line with the result in Theorem~\ref{thm:l1}, we observe that even with the entire healthy population social distancing both viruses are endemic, and the system converges to a coexistence equilibrium along the line $\mcl_1$; see Figure~\ref{figure:L_1:traj}.


\begin{figure}[h!]\centering
\begin{minipage}{0.45\textwidth}\centering
    \includegraphics[width=.9\textwidth]{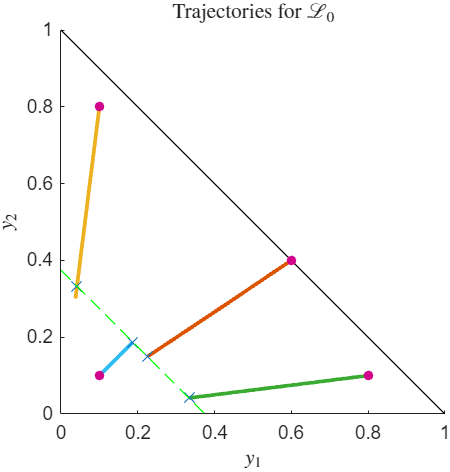}
    \caption{
    The dashed green line is a line of coexistence equilibria ($\mcl_{0}$) that is also locally stable.}
    \label{figure:L_0:traj}
\end{minipage}
\hspace{5mm}
\begin{minipage}{0.45\textwidth}\centering
    \includegraphics[width=0.9\textwidth]{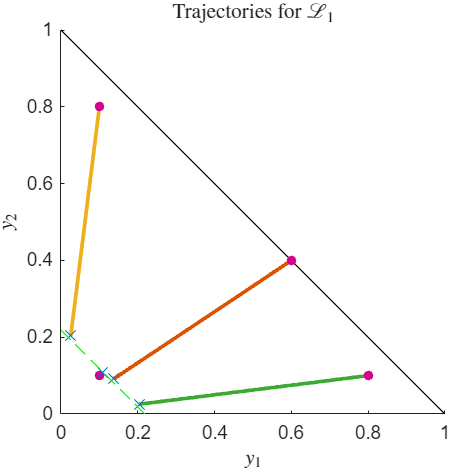}
    \caption{
    The dashed green line is a line of coexistence equilibria ($\mcl_{1}$) that is also  locally stable.}
    \label{figure:L_1:traj}
\end{minipage}
\end{figure}

\noindent\textbf{Coexistence equilibria with some SD:}
From the analysis of $\mcl_\S$, we get several constraints. With $R^0 := \frac{\beta_1}{\delta_1} = \frac{\beta_2}{\delta_2}$, $y_1$ must satisfy inequality~\eqref{eq:lineS}.
Recall the assumption $r_1 < r_2$ which implies $r_1 / r_2 < 1$.
This requires at least that
$\frac{c_\mcd}{2r_2} + \frac{1}{qR^0} < 1$.
It must also be that
$\beta_1\ell(y_1)q^2 < \delta_1 < \beta_1\ell(y_1)q$
so that $z_\S = \frac{1}{1-q} - \frac{\delta_1}{\beta_1\ell(y_1)q(1-q)} \in(0,1)$;
notice that if $q = 1$, i.e., social distancing has no effect, then the condition for $z_\S$ is unsatisfiable.
We note that our example---see Figure~\ref{figure:L_S}---appears to be locally stable with a small domain of attraction being sensitive to change in initial conditions.

\begin{figure}[h!]\centering
\includegraphics[width=0.65\textwidth]{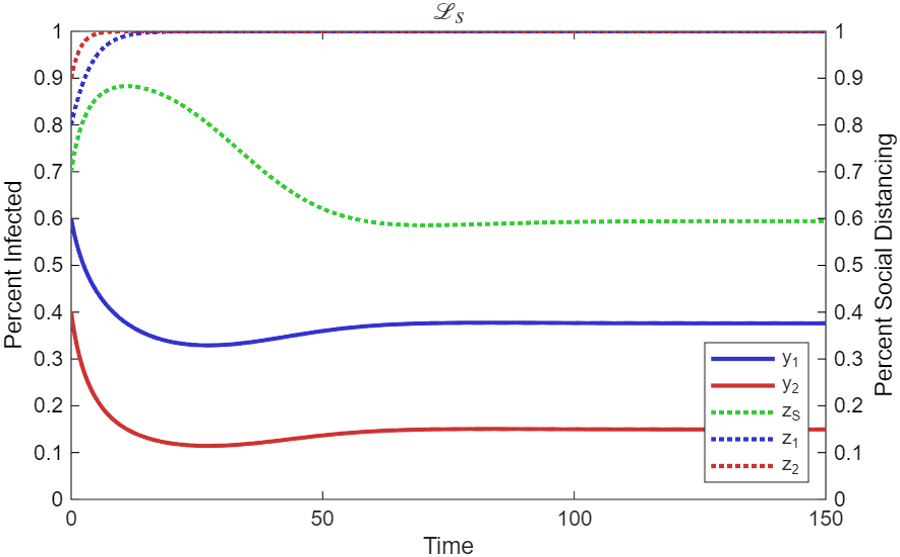}
\caption{Simulation 
showing both viruses endemic and partial social distancing in the healthy population, i.e., the system converging to $\mcl_\S$.}
\label{figure:L_S}
\end{figure}

\subsection{Effect of adding a co-ordination term to the dynamics}
\seb{Suppose instead that social distancing is socially reinforced by those participating. We modify equation~\eqref{eq:payoff:sn}, the payoff for not social distancing, as
\begin{align}
    \pi_{\S\mcn}(y) &= -ry + (1 - z_\S) \label{eq:payoff:sn:revised} \\
    \implies  \dot{z}_\S &= z_\S(1-z_\S)(2ry - c_\mcd + z_\S - 1)
\end{align}
Equation~\eqref{eq:payoff:sn:revised} contains the \emph{coordination term} $(1-z_S)$, which indicates that  the larger the  fraction of the susceptible population that is social  distancing (resp. not social distancing), the smaller the payoff for not social distancing (resp. social distancing). 
Therefore, the modified system has the following dynamics:
\begin{equation}\label{eq:main:revised}
    \begin{aligned}
        \dot{y}_1 &= y_1 (\beta_1 s (1 - z_\S(1-q))(1 - z_1(1-q)) - \delta_1) \\
        \dot{y}_2 &= y_2 (\beta_2 s (1 - z_\S(1-q))(1 - z_2(1-q)) - \delta_2) \\
        \dot{z}_\S &= z_\S(1-z_\S)(2ry - c_\mcd + z_\S - 1) \\
        \dot{z}_1 &= z_1(1-z_1)(c_1 - c_\mcd) \\
        \dot{z}_2 &= z_2(1-z_2)(c_2 - c_\mcd) 
    \end{aligned}
\end{equation}

Observe that, for system~\eqref{eq:main:revised}, the quantity $2ry - c_\mcd$ does not directly determine whether $z_\S$ is increasing or decreasing but rather the inflection point above which $z_\S$ is increasing and below which it is decreasing.
In this model, healthy individuals may not social distance even when the virus is prevalent, if few others are already social distancing ($z_\S(0) \ll 1$).
We choose the parameters
$\beta=[\begin{smallmatrix}
    0.8&&0.8
\end{smallmatrix}]$, $\delta=[\begin{smallmatrix}
    0.1&&0.1
\end{smallmatrix}]$, $r=[\begin{smallmatrix}
    1.1&&1.1
\end{smallmatrix}]$, $c=[\begin{smallmatrix}
    0.2&&0.3
\end{smallmatrix}]$, $c_{\mcd}=1$ and $q=0.1$.
Observe that $q \beta_i/\delta_i < 1 < \beta_i/\delta_i$ and $2r_i - c_\mcd - 1 \in(0,1)$.
Let $x(0) = [0.5, 0.3, 0.6, 0.1, 0.2]^\top\in\interior\Gamma$. In Figure~\ref{figure:coord_limit_cycle}, we can observe that the system appears to enters a limit cycle under such conditions.
If we change only $\beta_1 = 0.9$, then the reproduction numbers are no longer identical and the virus with the larger reproduction number has 
a limit cycle whose amplitude is increasing 
(see the blue line in Figure~\ref{figure:coord_limit_cycle_gen}), whereas the virus with the smaller reproduction number has \sebcancel{oscillations with decaying amplitude} 
(see the red line in Figure~\ref{figure:coord_limit_cycle_gen}). That is, the virus with the larger reproduction number dominates. 
\begin{figure}[h!]\centering
\begin{minipage}{0.45\textwidth}\centering
\includegraphics[width=0.9\textwidth]{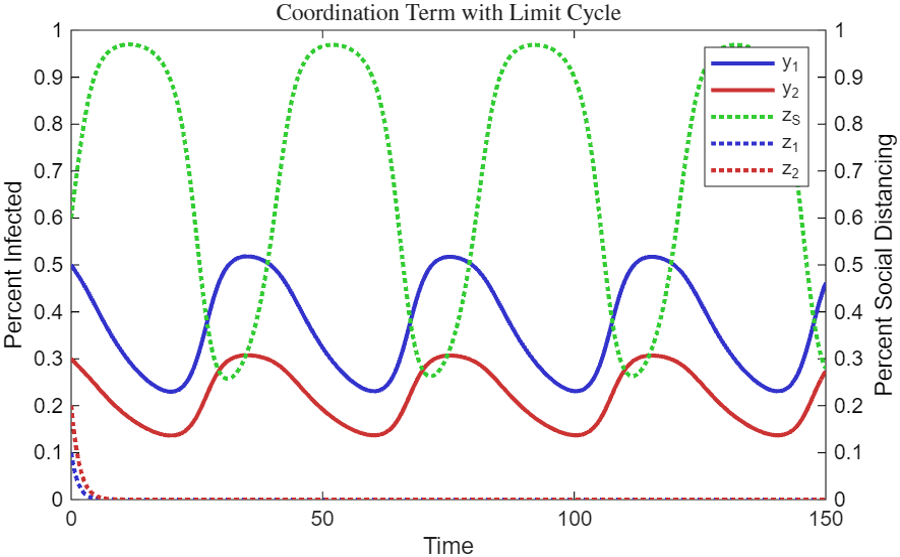}
\caption{Simulation of coordination term system with limit cycle and identical viral parameters.}
\label{figure:coord_limit_cycle}
\end{minipage}
\begin{minipage}{0.45\textwidth}\centering
\includegraphics[width=0.9\textwidth]{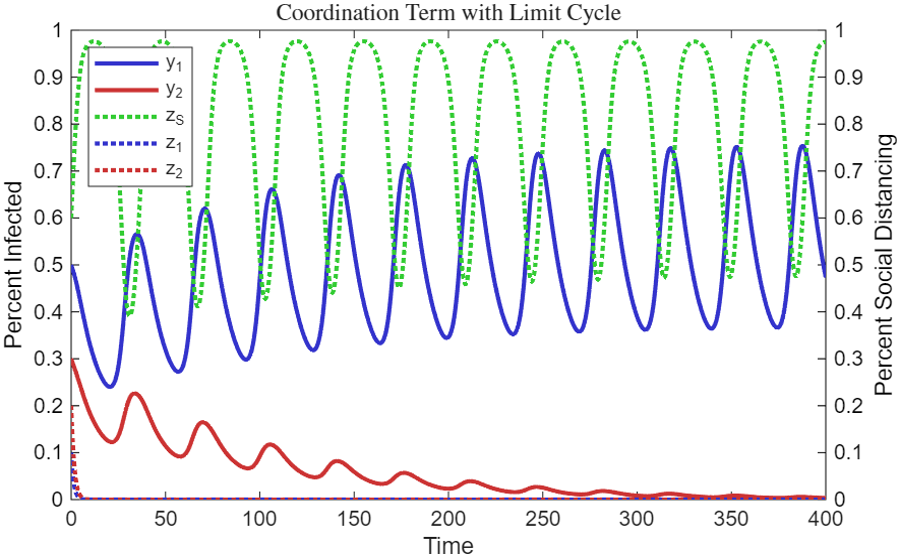}
\caption{Simulation of coordination term system with limit cycle and generic parameters.}
\label{figure:coord_limit_cycle_gen}
\end{minipage}
\end{figure}
We note that convergence 
is sensitive to initial conditions.
With the initial conditions
$x(0) \in\{[0.5, 0.3, 0.6, 0.1, 0.2]^\top,\allowbreak [0.1, 0.8, 0.9, 0.9, 0.9]^\top,\allowbreak  [0.1, 0.1, 0.5, 0.5, 0.5]^\top,\allowbreak [0.8, 0.1, 0.5, 0.5, 0.5]^\top\}$,
Figure~\ref{figure:coord_limit_cycle_traj} 
indicates different trajectories either converging to 
a unilateral 
\seb{equilibrium} or to the limit cycle, depending on their initial conditions. More specifically, the green dots denote the initial conditions, while the crosses denote 
\seb{equilibria} (assumed).
Figure~\ref{figure:coord_limit_cycle_traj} (a) clearly shows virus~2 dying out in all cases and virus~1 remaining endemic.
Figure~\ref{figure:coord_limit_cycle_traj} (b) illustrates the cyclic behavior of the blue and red trajectories, oscillating in $y_1$ and $z_\S$. Note that the yellow and purple trajectories appear to converge to a unilateral 
\seb{equilibrium} even from initial conditions far away, while the limit cycle appears to be attractive within a small domain.
}
\begin{figure}[h!]\centering
\begin{minipage}{0.45\textwidth}\centering
\includegraphics[width=0.9\textwidth]{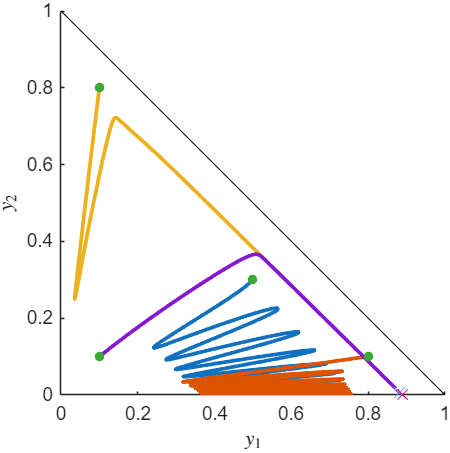}
\ (a)
\end{minipage}
\begin{minipage}{0.45\textwidth}\centering
\includegraphics[width=0.9\textwidth]{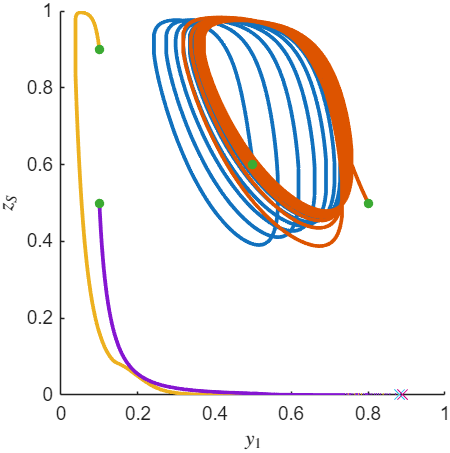}
\ (b)
\end{minipage}
\vspace{-3mm}
\caption{(a): The trajectories of system~\eqref{eq:main:revised} projected onto the $y_1$-$y_2$ plane; (b): The trajectories of system~\eqref{eq:main:revised} projected onto the $y_1$-$z_\S$ plane. Trajectories sufficiently near the limit cycle converge to the limit cycle, while trajectories away from the limit cycle converge to a unilateral 
\seb{equilibrium} instead (virus~1 being the dominant virus with $R_1>R_2$).}
\label{figure:coord_limit_cycle_traj}
\end{figure}

\section{Conclusion}\label{sec:conclusions}

We have presented and analyzed a novel competitive bi-virus model with mediation by social distancing, where social distancing behavior is dynamic.
We showed that the proposed model is not monotone, unlike the classic bi-virus model; consequently, 
understanding the limiting behavior of our model is far more challenging, since one cannot leverage the rich literature on MDS.
Next, by assuming that the costs of taking social distancing measures and not social distancing (i.e., the cost of putting other individuals at risk) were related arbitrarily, we 
listed the various DFEs that our model admits, and identified conditions for almost global stability (resp. instability) of each DFE. Thereafter, in order for ease of exposition, we assumed that the aforementioned costs followed a certain relation, and identified sufficient conditions for almost global stability (resp. instability) of the unilateral (single virus endemic) equilibria, as well as local stability (resp. instability) for 
continuua of coexistence equilibria. 

Several avenues present themselves for future work. First, we could seek to identify conditions for global stability of lines of coexistence equilibria. 
Second, as mentioned previously, one could aim to understand the limiting behavior of our model, possibly by leveraging the theory of singular perturbations. Third, it would be interesting to see if the combination of  
replicator equations which allow for periodic opinion dynamics with the classical bi-virus SIS model has the potential to introduce limit cycles or otherwise complicate the behavior of the model.
Fourth, observe that our model is a scalar model, i.e., restricted to the single-node case; embedding a network into our model might potentially introduce new behavior. 
Finally, opinion dynamics could be modeled in an additional network, which need not 
be identical to the viral transmission network. Such models could examine how local versus national opinions affect responses to disease.

\section*{Appendix}

\newenvironment{reading}{\begin{@abssec}{Further Reading}}{\end{@abssec}}
\begin{reading}
For game theory, see \cite{sandholm2010population,hofbauer2007evolution,osborne2004introduction,reny2020nash}.\\
For competitive viruses from a dynamical systems theory viewpoint, see \cite{sahneh2014competitive,ye2022competitive,liu2019analysis,pare2021multi,ye2021convergence,anderson2023equilibria,janson2024competitive}.\\
For competitive viruses from a biology viewpoint, see \cite{laurie2018evidence,wu2020interference,pepin2008asymmetric,poland1996two,remuzzi2020covid,heesterbeek2015modeling,Chan2018,funk2009spread}.
\end{reading}

\newenvironment{contributions}{\begin{@abssec}{Author declarations and contributions}}{\end{@abssec}}
\begin{contributions}
The authors declare that they have no conflicts of interest.\\
KP and SG conceptualized the work. BC and KP performed all mathematical analysis, while SG supervised BC's work. BC ran all the numerical simulations.
All authors contributed towards preparing the original draft of the manuscript.
\end{contributions}

\newenvironment{data}{\begin{@abssec}{Data availability}}{\end{@abssec}}
\begin{data}
Code for reproducing the presented simulations is available at
\url{https://github.com/BenCat8472/Bi-SIS-Game}.
\end{data}

\bibliographystyle{siamplain}
\bibliography{keith,references}
\end{document}